\documentclass[10pt,twocolumn,twoside]{IEEEtran}
\usepackage{cite}
\usepackage{graphicx,subcaption,color,mathrsfs,dsfont}
\usepackage[font={small}]{caption}
\usepackage[ruled,vlined,english,onelanguage]{algorithm2e}
\usepackage{amsmath,amssymb,tikz,upgreek}
\usepackage{algorithmic}
\usepackage{pgfplots}	
\usepackage{tikzscale}
\usepackage[active]{srcltx}
\usepackage{nomencl}
\usepackage{ifthen}

\allowdisplaybreaks
\DeclareMathOperator*{\argmax}{arg\,max}
\DeclareMathOperator*{\argmin}{arg\,min}

\newcommand{\BlackBox}{\rule{1.5ex}{1.5ex}}  
\newenvironment{proof}{\par\noindent{\bf Proof\
}}{\hfill\BlackBox\\[2mm]}

\newtheorem{lemma}{\bf{Lemma}}

\newtheorem{remark}{\bf{Remark}}

\newtheorem{assumption}{\bf{Assumption}}

\newcommand{\be}{\begin{equation}}
\newcommand{\ee}{\end{equation}}
\newcommand{\bea}{\begin{eqnarray*}}
\newcommand{\eea}{\end{eqnarray*}}

\makeatletter
\def\Lddots{\mathinner{\mkern1mu\raise17\p@\vbox{\kern17\p@\hbox{.}}\mkern2mu
    \raise8\p@\hbox{.}\mkern2mu\raise\p@\hbox{.}\mkern1mu}}
 \makeatother

\outer\def\subsect#1\par{\vskip12pt
plus.07\vsize\penalty-250\vskip0pt plus-.07\vsize
\bigskip\vskip\parskip\message{#1}
\vbox{\smash{\lower9pt\hbox{\kern-8pt\epsfbox{shadedbox.eps}}}}\vskip-\baselineskip
\leftline{\large\bf#1}\nobreak\medskip}

\def\BibTeX{{\rmfamily B\kern-.05em{\scshape i\kern-.025em b}\kern-.08em \TeX}}

\newcommand{\A}{\mathcal{A}}
\newcommand{\cB}{\mathcal{B}}
\newcommand{\C}{\mathcal{C}}
\newcommand{\E}{\mathcal{E}}

\newcommand{\N}{\mathcal{N}}

\newcommand{\bs}{\mathbf{s}}

\newcommand{\bb}{\mathbf{b}}

\newcommand{\bdelta}{\mbox{\boldmath $\delta$}}
\newcommand{\bgamma}{\mbox{\boldmath $\gamma$}}
\newcommand{\bomega}{\mbox{\boldmath $\omega$}}

\makeatletter
\newcommand*{\coloneq}{\mathrel{\rlap{%
                     \raisebox{0.3ex}{$\m@th\cdot$}}%
                     \raisebox{-0.3ex}{$\m@th\cdot$}}%
                     =}

\renewcommand{\nomgroup}[1]{
 \ifthenelse{\equal{#1}{V}}{}{
 \ifthenelse{\equal{#1}{S}}{\item[List of Symbols]}{}}
}

\makenomenclature

\begin{document}

\title{An Active Sequential Xampling\\ 
Detector for Spectrum Sensing}

\author{Lorenzo Ferrari, and Anna Scaglione
\IEEEcompsocitemizethanks{\IEEEcompsocthanksitem L. Ferrari and A. Scaglione \{Lorenzo.Ferrari,Anna.Scaglione@asu.edu\} are with the School of Electrical, Computer and Energy Engineering, Arizona State University, Tempe,
AZ, 85281, USA.\protect
}
\thanks{This work was supported in part by the National Science Foundation CCF-247896 grant.}}
\maketitle
\IEEEpeerreviewmaketitle
\begin{abstract}\color{black}
Nyquist sampling at very high carriers can be prohibitively costly for low power wireless devices. In spectrum sensing, this limit calls for an analog front-end that can sweep different bands quickly, in order to use the available spectrum opportunistically.
In this paper we propose a new sub-Nyquist analog front-end and a sensing strategy formulated as a sequential utility optimization problem. The sensing action maximizes a utility function decreasing linearly with the number of measurements, and increasing as a function of the active spectrum components that are correctly detected. The optimization selects the best linear combinations of sub-bands to mix, in order to accrue maximum utility. 
The structure of the utility represents the trade-off between exploration, exploitation and risk of making an error, that is characteristic of the spectrum sensing problem. 
We first present the analog front-end architecture, and then map the measurement model into the abstract optimization problem proposed, and analyzed, in the remainder of the paper. 
We characterize the optimal policy, under constraints on the sensing matrix, and derive the approximation factor of the greedy approach we introduce to solve the problem. 
Numerical simulations showcase the benefits of combining active spectrum sensing with sub-Nyquist sampling.
\end{abstract}
\section{Introduction}
\label{sec:introduction}
{\color{black}
Mobile devices are increasingly being used for entertainment (gaming, video streaming): the coexistence of these new services with the ``Internet of Things'' (IoT) and Machine-to-Machine (M2M) communications means that wireless applications may quickly become starved for bandwidth. 
Millimeter wave bands can provide the much needed linear increase in throughput, but pose the challenge of high speed sampling, required to sense the spectrum, for many 
relatively low rate IoT applications, which are likely to benefit from opportunistic decentralized spectrum access.
 In order to achieve efficient usage of the spectrum, when a large portion is potentially available, we need to overcome the bottleneck of Nyquist sampling, which can be prohibitive, especially for low power wireless devices.

The literature offers two classes of solutions to eliminate this bottleneck: a static sub-Nyquist 
sampling front-end or a tunable narrowband coupled with active sensing strategies. 
\par The static sub-Nyquist sampling approach is based on the representation of the signal as a nonlinear Union of Subspaces (UoS) \cite{lu2008theory}.
A common framework to cover several acquisition and reconstruction approaches under the umbrella of the UoS model was defined in \cite{mishali2011xampling}, which named these analog to digital conversion techniques {\it Xampling architectures}. 
Xampling architectures preprocess the signal in the analog domain, and then sample at a lower rate compared to what the Nyquist theorem dictates.
The aim is to reduce the complexity and energy cost for the Analog-to-Digital Converter (ADC) hardware.
The downside is the increased complexity in the reconstruction of the underlying signal.
Our work focuses on {\it multiband signals}, whose UoS representation is a {\it finite} union of subspaces with {\it infinite}, but countable dimensions, in the space spanned by orthogonal sinc functions.
Examples of Xampling architectures for multiband signals are in e.g.  \cite{fudge2008nyquist,mishali2009blind,mishali2010theory,venkataramani2000perfect}.
There are other Compressive Spectrum Sensing (CSS) algorithms in the literature that are related. Typically, they start directly from a discrete time model (see e.g. \cite{zeng2011distributed,zeinalkhani2012iterative}) where the receiver has a fixed number of measurements, forming an underdetermined system of equations, whose solution is a sparse vector with support equal to the spectrum occupancy.

The second set of approaches is fully adaptive and consists in selecting opportunistically, and in a cognitive fashion, a small section of the spectrum at a time, relying on an analog front-end able to swiftly switch between small sub-bands. The problem choosing optimally the band to explore was studied by several authors, see e.g. \cite{zhao2008myopic,unnikrishnan2010algorithms,wang2012optimality} and our previous work \cite{ferrari2017utility} that also focuses on single band tests. Note that tests that alias the spectrum create correlation in the resources, and this is why our formulation departs from the Partially Observable Markov Decision Processes (POMDPs) models used in previous works for single band tests.
More recently, other authors have proposed and studied sensing strategies that would reduce the required number of CSS measurements to the sparsity of the signal. The work in \cite{michelusi2015cross} 
has a predetermined sensing phase duration, while \cite{ma2017sparsity} introduces an analog preprocessing which guarantees satisfactory detection performance, irrespective of the sparsity of the signal, as long as the measurement phase is proportional to the average occupancy. 
Our contribution and motivation is discussed next.
\subsection{Contributions}
Our main contribution is a new dynamic optimization framework married with the design of the active sub-Nyquist spectrum sensing frontend.  

We provide two arguments to study alternatives to Xampling (or to the aforementioned CSS methods) 
previously studied. 
First, in the spectrum sensing problem, the objective is the detection of the idle channels, not the signal reconstruction: this suggests that the Xampling complexity may still exceed what is really necessary for this task, as previously discussed in \cite{cohen2014sub}.
Additionally, most of the standard results in Compressive Sensing (CS), that bound the $\ell_2$-norm of the estimation error, do not directly express the detection performance.  
The architecture studied in this paper has the advantage of being sequential, requiring incoherent observations and being robust to time inaccuracies in the sampling hardware, as opposed to e.g. the multi-coset approach in \cite{venkataramani2000perfect}. 
For the spectrum sensing detection problem, the additive noise at the receiver plays an important role in the performance of interest. Hence, rather than focusing on reconstruction in {\it noiseless} scenarios, in this paper we directly tackle the so called {\it noise folding} problem in the design \cite{arias2011noise}. 
Noise folding gives a Signal-to-Noise Ratio (SNR) deterioration approximately linear in the number of bands that are aliased prior to sampling \cite{arias2011noise}. This can cause poor performance for several Xampling approaches at low SNR.
As discussed in \cite{baron2010bayesian}, low density measurement matrices represent an effective countermeasure to noise folding.
Additionally sparse matrices enable belief propagation techniques (i.e. message passing) for signal recovery (or detection in our context) that lead to state of the art performance, in spite of the poor conditioning of the sensing matrices. 

Second, to the best of our knowledge, receivers in the Xampling family that are in the literature (see e.g. \cite{fudge2008nyquist,mishali2009blind,mishali2010theory,venkataramani2000perfect,lexa2011compressive,ariananda2012compressive,cohen2014sub}) do not fall in an {\it active sensing framework}, as they are not married to an online optimization of the measurement strategy. Hence, that can have a  limitation: if the spectrum is not sufficiently sparse, neither the signal reconstruction, nor the detection of its presence in a certain band, can be accurate, even in the absence of noise. 
In fact, for general non-sparse signals, in a {\it noiseless} setting, \cite{cohen2014sub} proved that half the Nyquist rate is necessary (see also \cite{lexa2011compressive, ariananda2012compressive} for related discussions).

In contrast to Xampling architectures used for multi-band signals (e.g. the Modulated Wideband Converter (MWC) in \cite{mishali2010theory}, the Multirate Asynchronous Sub-Nyquist Sampler (MASS) in \cite{sun2012wideband} and also \cite{saeed_restless}), our CSS sequential architecture only needs a single mixer, with a programmable analog waveform and Low Pass Filter (LPF).
In our optimization framework, the Cognitive Receiver (CR) needs to select what group of  sub-bands (generally non-contiguous) to sense at each test. A test corresponds to a sample, obtained after folding different sub-bands of the wide-spectrum signal.
Compared to other multi-band signals receivers, the hardware in our architecture is simpler, since we use a single non-coherent receiver and a single sampling device that 
collects energy measurements sequentially, sampling at a fraction of the Nyquist rate.
 We use a time-dependent utility function to optimize the trade-off between sensing and exploitation.

It is important to remark that, similarly to \cite{zhao2008myopic,unnikrishnan2010algorithms,wang2012optimality,ferrari2017utility}  and \cite{saeed_restless},  the optimum action will not generally attempt the full recovery of all the white spaces. In fact, the optimum decision may be conservative and sense a very limited portion of the spectrum. Furthermore, since scanning one sub-band at a time is a possible action of our active spectrum sensing strategy, our method subsumes previous techniques to scan the spectrum optimally, without mixing it. We refer to this approach as that of {\it Direct Inspection (DI)} and discuss it in Section \ref{subsec:DI}. 

Our work is more closely related with the stochastic optimization schemes that extend the framework in  \cite{zhao2008myopic,unnikrishnan2010algorithms,wang2012optimality}, and optimize a CSS action based on previous observations \cite{hao2012sequential,malioutov2010sequential,zhao2015cooperative,braun2015info,saeed_restless}. With the exception of \cite{saeed_restless}, the common goal in these papers is the recovery of the full support of a given vector. Typically, the techniques proposed are shown to be able to cope with lower SNR in the signal reconstruction with low complexity.
What these optimizations do not capture is the fact that, in cognitive spectrum sensing applications, a timely decision is also desirable, to have enough time to exploit the spectrum. }
In fact, our method is also adaptive with respect to the time horizon $K$, the number of resources $N$, the prior probabilities on the states of each resources, and the parameters that characterize the utility function (i.e. reward/penalty for good/bad decisions). 
Interestingly, from our performance analysis, it is clear that sparse and adaptive sensing matrix designs outperform dense sensing matrices, as well as those that are sparse, but static. There are two main reasons for this: 1) through belief propagation algorithms, they achieve near-optimum detection performance; 2) they mitigate the aforementioned noise folding phenomena.  
We also emphasize in the paper that our model is applicable not only when the utility comes from finding empty entries (e.g. spectrum sensing), but also when one is interested in finding the occupied ones (e.g. in a RADAR application).

More specifically, the paper is organized as follows: Section \ref{sec:model} is dedicated to the signal model and the analog front-end of our detector, and in Section \ref{sec:opt_problem} we formulate the optimization problem. Then, in Section \ref{sec:design} we study the optimal dynamic design, for the Direct Inspection (DI) case (\ref{subsec:DI}), where there is no mixing of sub-bands (also known as scanning receiver), and a Group Testing (GT) case (\ref{subsec:GT}), where we introduce the possibility of mixing different bands. We will show that, even if finding the optimal policy is exponentially complex in the number of resources, we can characterize the approximation factor for a greedy procedure. Section \ref{sec:additional} is dedicated to alternative detection approaches: linearization of Maximum Likelihood (ML) estimate and covariance estimation via LASSO relaxation. Numerical results to sustain our claims are presented in Section \ref{sec:simulations}.\\[.2cm] 
{\bf Notation}
We use bold lower-case to represent vectors, bold upper-case for matrices and calligraphic letters to indicate sets.
With $\bs_{\A}$ we indicate the entries $i\in\A$ of vector $\bs$, and with $\|\mathbf{y}\|_{\mathbf{A}}^2$ we represent the weighted $\ell_2$-norm  $\mathbf{y}^T\mathbf{A}\mathbf{y}$.
For any set function $f(\A)$ we define the marginal increment for adding 
element $a$, as $\partial_a f(\A)=f(\A+a)-f(\A).$ 
\section{Xampling Detector}\label{sec:model}
In the context of spectrum sensing for cognitive radio, in addition to the payload, each transmission includes large amounts of control signals overhead.
It is then natural to assume that the activity of the Primary Users (PUs), in a certain spectrum, will persist for several sampling periods (see Fig.\ref{fig:cognitive_figure}).
\begin{figure}[ht]
\centering
\includegraphics[width=\linewidth]{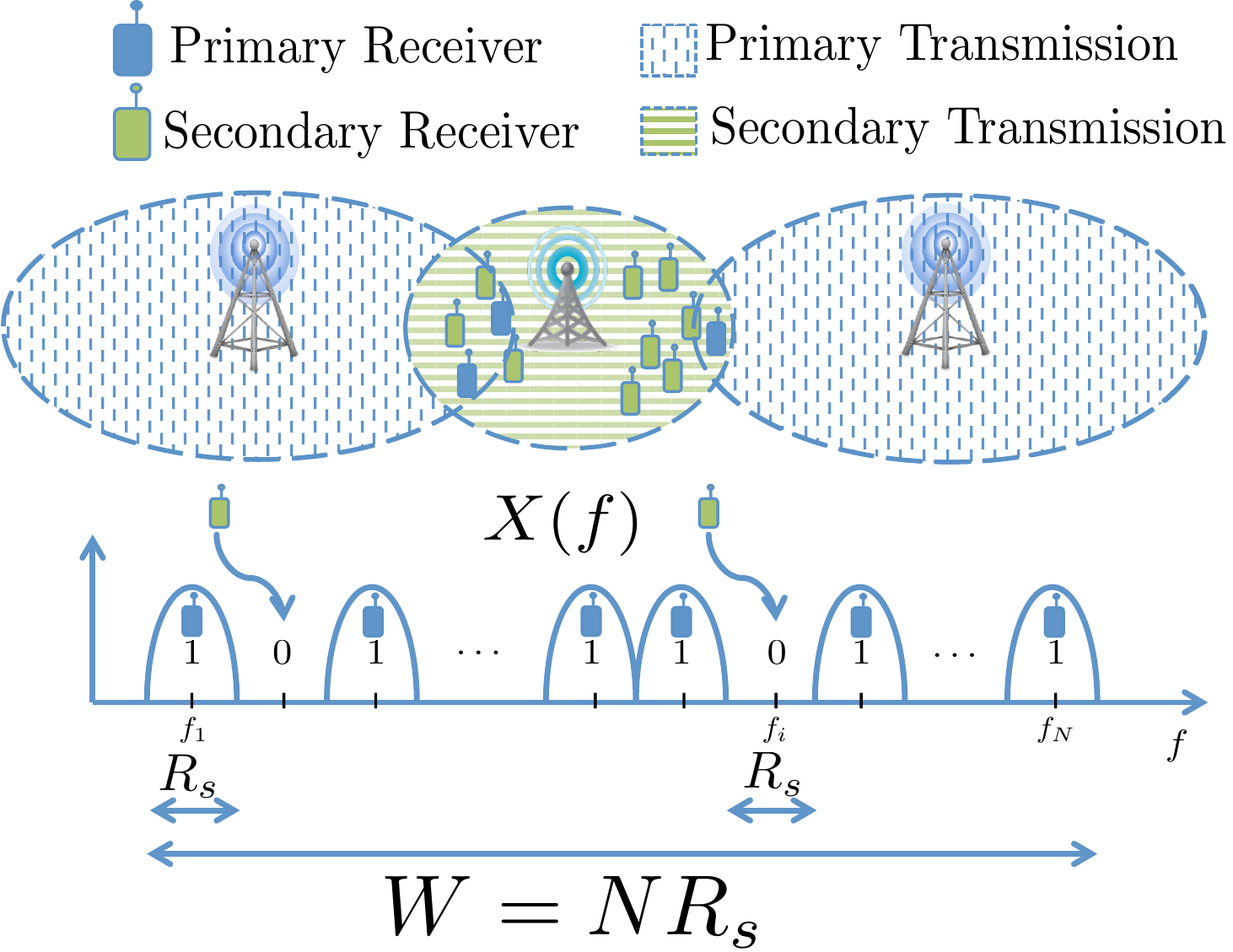}
\caption{Cognitive radio scenario}\label{fig:cognitive_figure}
\end{figure} 
However, assuming this interval lasts $T=KT_s$, the sensing {\color{black} mechanism should provide} the fastest decision it can. 
The goal of the proposed cognitive receiver is to sequentially sense the spectrum for the first portion of the interval and transmit the most it can during the remaining time, {\color{black} over the sub-channels found empty.}
\subsection{Sensing Architecture and Observation Model}
We assume that the complex envelope of the analog signal we are exploring is a multicomponent signal, with overall bandwidth equal to $W=NR_s$.
The components are indexed by $i\in\N$, where $\N\triangleq\{1,2,\dots,N\}$.
The elements of the binary vector  $\bs=\{s_i:i\in\N\}\in\{0,1\}^N$  indicate presence (1) or absence (0) of a component (i.e. Primary User (PU) communication signal) over the $i$-th sub-band.
During the interval $0\leq t < T=KT_s$ the received signal is:
\begin{align}\label{eq.model}
	y(t) &= x(t)+w(t)\\
	x(t) &= \sum_{i=1}^{N} s_i
 x_i(t) e^{-j2\pi R_s (i-1)  t}.
\end{align} 
with $w(t)\sim {\cal N}(0,N_0\delta(\tau))$  being Additive White Gaussian Noise {\color{black} (AWGN)}.
The components of the received signals $x_i(t)$ correspond to each PU source, modeled as band-limited random processes with bandwidth $R_s$; they are equal in the mean square sense to the following process: 
\be\label{eq:sum_x_i}
x_i(t)=\sum_{k=1}^{K}x_i[k]sinc(\pi (R_st-k+1)).
\ee
Rather than having a filter bank architecture as in \cite{saeed_restless}, to further reduce the hardware complexity, we base our scheme on a sequential non-coherent sampling architecture. 
\begin{figure}[!htbp]
\centering
\includegraphics[width=\linewidth]{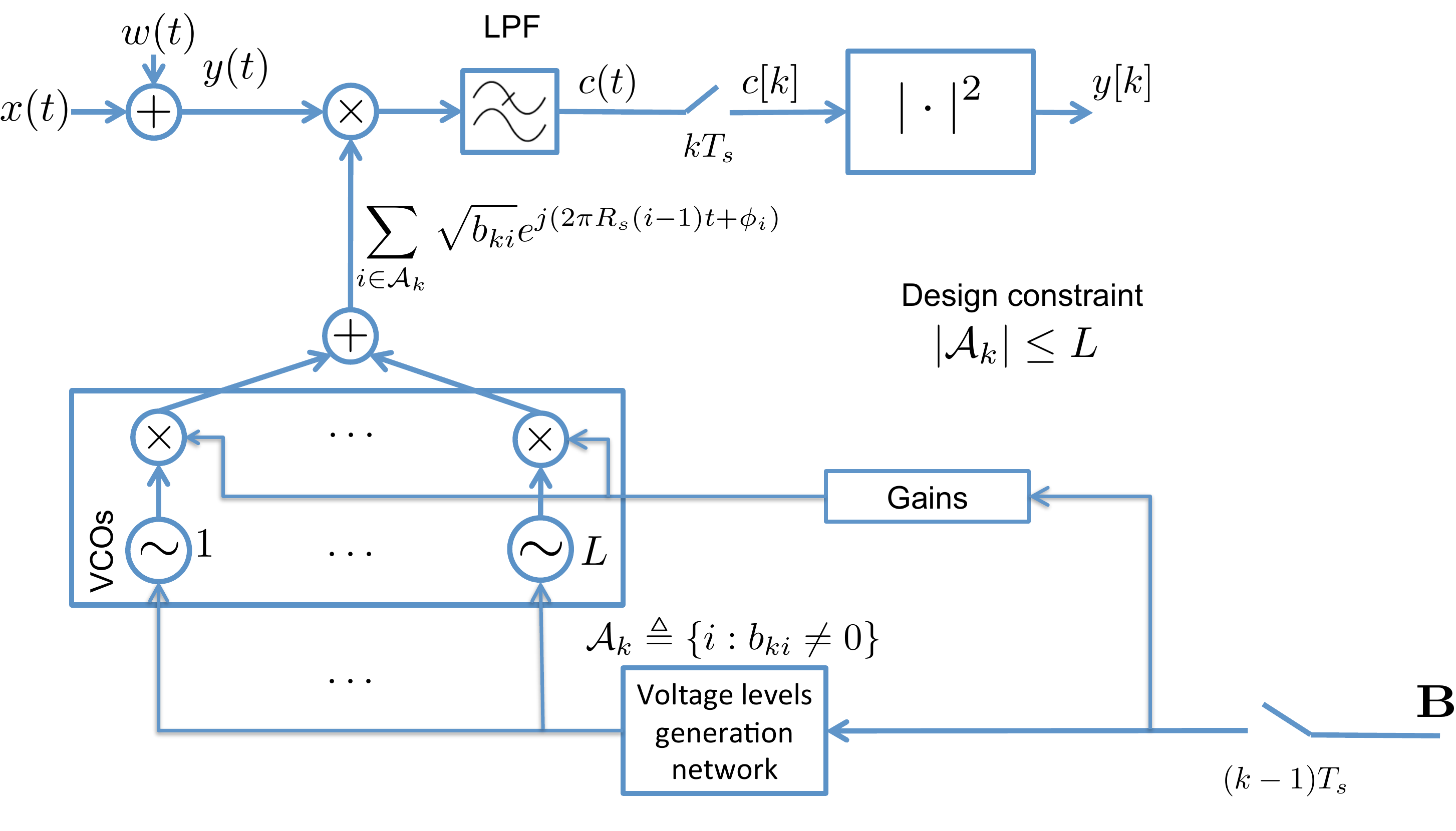}
\caption{Sequential Xampling detector hardware architecture. The inputs of the $L$ Voltage Controlled Oscillators (VCOs) are voltage levels proportional to the position of the non-zero entries of the vector $\mathbf{b}_k$, to switch and recenter their frequency at the corresponding multiple of $R_s$, according to the prescribed sensing action.}
\label{fig:circuit_diagram}
\end{figure}
As the diagram in Fig. \ref{fig:circuit_diagram} shows, the sequential receiver we propose first modulates the received signal 
over the period $(k-1)T_s\leq t <kT_s$ with a signal $g_{k}(t)$, synthesized by 
selecting an appropriate input to the $L$ VCOs. More specifically, if 
we denote by ${\cal A}_k\subset {\cal N}$ the subset of $|{\cal A}_k|\leq L$
sub-band mixed in the sample for the $k$th test, we have:
\begin{align}
g_{k}(t)&=\sum_{i=1}^N a_{ki}e^{j 2\pi R_s (i-1) t},\label{eq:mixing_signal}\\
a_{ki}&=\sum_{u\in {\cal A}_k}\sqrt{b_{ku}}e^{j\phi_{ku}}\delta[i-u]\label{eq:a_i_def}. 
\end{align}
where, as shown in \eqref{eq:a_i_def}, only $|{\cal A}_k|\leq L$ tones are activated, and the phase $\phi_{ku}$ in \eqref{eq:mixing_signal}-\eqref{eq:a_i_def} accounts for the delay in generating the tone at the $u$-th frequency included in ${\cal A}_k$ plus the oscillator phase, while $\sqrt{b_{ki}}$ is the amplitude. The optimal choice of the set ${\cal A}_k$ is the subject of Section \ref{sec:opt_problem}.
Correspondingly, for our incoherent detector, the sensing vector coefficients $b_{ki}$ (that is, the $k$th row of the sensing matrix ${\mathbf B}$), associated to the $k$th sample, are:
\be
b_{ki}=|a_{ki}|^2=\sum_{u\in {\cal A}_k}b_{ku}\delta[i-u]\label{eq:bki}
\ee
irrespective of the random phase $\phi_{ku}$ of the VCO tones activated. Clearly, if $L\ll N$, ${\mathbf B}$ is a sparse matrix.
Then, after convolving the modulated signal with an ideal low-pass filter, with impulse response $sinc(\pi R_s t)$, the receiver samples the output $c(t)$ at times $kT_s,~k=1,\ldots,\kappa$. This operation is equivalent\footnote{If the periodic signals where not truncated in time the relationship would be exact, in practice there will be some approximation error due to the windowing of the signal over the prescribed interval $[(k-1)T_s,kT_s]$. The effect of this can be mitigated by using raised cosine filtering and a non rectangular window to reduce the effect of side lobes. } to an orthogonal projection, as shown below:
\begin{align}
c[k]&=\left.[y(t)g_{k}(t)] \star R_s sinc(\pi R_s t)\right|_{t=kT_s}=\sum_{i=1}^N\sqrt{b_{ki}}e^{j\phi_i}Y_{ki}
\end{align}
where $Y_{ki}$ represent the orthogonal projections over the period  $(k-1)T_s\leq t <kT_s$ of $y(t)$ over the following signals:
\be
Y_{ki}= \langle y(t),R_se^{j2\pi R_s (i-1)  t}sinc\left(\pi \left(R_st-k+1\right)\right)\rangle.
\ee
Note that the signals 
$$\left\{e^{j2\pi R_s (i-1)  t}sinc\left(\pi \left(R_s t-k+1\right)\right)\right\}_{i,k\in \mathbb{Z}}$$ 
form a orthogonal basis, and that \eqref{eq.model} is equivalent to:
\begin{align}\label{eq.model-expanded}
	x(t) &= \!\!\sum_{k=1}^{K}\sum_{i=1}^Ns_ix_i[k]
		 e^{j2\pi R_s (i-1)  t}sinc(\pi (R_st-k+1))
\end{align} 
\be
Y_{ki}=s_ix_i[k]+w_i[k]
\ee
where $w_i[k]\sim {\cal C N}(0,n_{i})$. If we model $x_i[k]$ also as i.i.d.
$x_i[k]\sim {\cal C N}(0,\varphi_i)$, we get that for a given state $\bs$:
\be
Y_{ki}\sim {\cal CN}(0,s_i\varphi_i+n_{i}).
\ee
where $\boldsymbol{\varphi}$ is a vector collecting the average, unknown {\it a priori}, received signal power from the existing communications.
The receiver samples for $k=1,\ldots,\kappa$ are:
\begin{align}
c[k]&=\sum_{i=1}^N\sqrt{b_{ki}}e^{j\phi_i}\left(s_ix_i[k]+w_i[k]\right)
\end{align} 
and therefore, assuming $\phi_i$'s are independent and uniformly distributed in $[0,2\pi)$, they are also conditionally zero mean Gaussian random variables:
\begin{align}
c[k]&\sim {\cal CN}\left(0,\theta[k]\right),\\
\theta[k]&=\theta(\bb_k,\bs)\triangleq\bb_k^T\left(\boldsymbol{\varphi}+\boldsymbol{n}\right).\label{eq:theta_def}
\end{align}
It follows that the information for the detection of the PU communications is embedded in the variance of the sample (which is the energy received during the $k$-th period). Sufficient statistics for our problem are:
\be\label{eq:y_def}
y[k]\triangleq |c[k]|^2
\ee 
which are exponentially distributed, i.e. $y[k]\sim Exp(\theta[k])$.
\footnote{Throughout the paper we will use, for convenience, the alternative parameterization for the exponential distribution, e.g. $f(k;\theta)=\frac{1}{\theta}e^{-\frac{x}{\theta}}u(x)$.}
\begin{remark}
{\it 
The signals are subject to linear distortion due to a multi-path channel. For the case of Rayleigh narrowband fading and digital modulations, such as PSK with constant amplitudes, the model is correct, but it is only an approximation in other cases. 
Also, it {\color{black} would be} appropriate to include a certain correlation among the samples  $x_i[k]$ in the case of a frequency selective channel.  We do not consider it, given the generalization is straightforward and the opportunistic strategy is a viable sub-optimum solution for that case as well. Finally it is useful to remark that, irrespective of the statistics of the received signal from the PUs, the variance of the samples 
$c[k]$ will have the same expression, which makes the energy detector a viable heuristic in general.   
}
\end{remark}

Our work will investigate the design of:
\begin{enumerate}
\item the $\kappa\times N$ measurement matrix $\mathbf{B}$, whose rows are the vectors $\bb_k$ (exploration phase)
\item a set of $N$ decision rules $\bdelta=\{\delta_i\in\{0,1\}:i=1,2,\dots,N\}$ over the unknown states $s_i$ of the resources (at the end of the exploration phase).
\end{enumerate}
Notice that the design of $\mathbf{B}$ includes:
\begin{itemize}
\item the measurement vectors $\bb_k$ for each test at time $k=1,2,\dots,\kappa$  (matrix rows),
\item  the sensing (exploration) time $\kappa$ to acquire information on the states $s_i$ via the observations $y[k]$ (number of rows).
\end{itemize}
\subsection{Hardware Considerations}\label{subsec:hard_limit}


From the sufficient statistics $y[k]$ in \eqref{eq:y_def}, it follows that our incoherent receiver can be implemented by combining different oscillators that do not require to be synchronized or phase-locked (see Fig.\ref{fig:circuit_diagram}). Furthermore, the switching frequency of the oscillator is also $R_s$, i.e. the single channel bandwidth, which is much slower than the limit of state-of-the art hardware. 
In the proposed receiver, the input is mixed with the signal $g_k(t)$, that folds the spectrum present in specific sub-bands onto the center frequency of the receiver, during what we can refer to as a {\it sub-Nyquist} carrier sensing phase. 
The samples are spaced by intervals of duration $T_{s}=1/R_s$ which is a factor $1/N$ smaller than the total spectrum.

Naturally, the mapping of the signal in general will be imperfect and, like in any ADC, calibration is necessary \cite{chen2010calibration,israeli2014hardware}. For most ADCs the assumption is that this calibration is done during an initial training phase, in which a known input signal can be used to estimate the equivalent matrix $\mathbf{B}$. 

As far as the proposed architecture is concerned, the circuit diagram of Fig.\ref{fig:circuit_diagram} assumes a {\it settling time} for the VCOs much smaller than $T_s$, i.e. the sampling period for the single channel sub-band. 
If this assumption does not hold, one should use a LPF with a smaller bandwidth and collect the samples $c[k]$ at an even slower rate than $R_s$, to wait for the VCOs to settle. This modification would not alter the statistical characterization of the samples, derived in the previous subsection.
The drawback of taking samples less often is that (assuming the same occupancy coherence time) one would have accrued less information than what is available in the received signal, and would have less than $K$ slots to decide. Given that our strategy is derived as a function of $K$, this would not invalidate our findings.
Another possibility would be to replace the $L$ tunable VCOs with $N$ oscillators at constant frequencies, corresponding to the $N$ possible bands of the signal. Using $N$ oscillators would increase the power consumption and cost of the circuit but would significantly reduce the switching time between two measurements.
Hence, this would be the natural choice if one wants to exploit a dense sensing matrix. Instead, the use of a bank of VCOs is preferable if the matrices are sparse because a small number of VCOs can synthesize the mixing signal.  
The switching would be in fact performed by a multiplexer, that would take the sum of the up to $L$ tones selected by the vector $\bb_k$.

In general, since we focus on the detection of the signal, with reasonably good device components we expect that calibration will either be far less demanding or unnecessary, if one accepts loss in sensitivity. In fact, the binary coefficients for the vector $\bb$ can be set to ones and zeros, as discussed in \ref{subsubsec:L_2}.  Controlling the gains is unnecessary for the system to work, and it is preferable to not add tunable gains, as they can be another possible source of uncertainty and complexity in the system. Finally, imperfect tuning of the VCOs will reduce the SNR, either by spreading or misplacing the center frequency of the components of interest, but not fundamentally impair the detection.
\section{Optimization framework}\label{sec:opt_problem}

During the times devoted to sensing $k=1,2,\dots,\kappa<K$ the player has the possibility to dynamically and adaptively design each measurement, by selecting a sensing vector $\bb_k=[b_{k1},b_{k2},\dots,b_{kN}]$, that indicates what subset of the entries of $\bs$ to probe (that is, 
the $b_{ki}$'s  will be non-zero only on the channels that are actively sensed).
To capture the optimization between sensing (exploration) and exploitation of a subset of the $\mathcal{N}$ sub-bands, we model a total utility proportional to the time left for exploitation $(K-\kappa)$. 
The detector acquires information about the entries $s_i$ via  observations with a p.d.f. parameterized by an unknown vector. 
Together with the optimal design of $\mathbf{B}$, the detector also optimizes the binary decision vector $\bdelta$, once the observations have been collected, on the state of the sub-bands.
We denote the type I (false alarm) and type II (miss) errors probability with $\alpha_i,\beta_i$ respectively, i.e. $\alpha_i=P(\delta_i=1|s_i=0)$ and $\beta_i=P(\delta_i=0|s_i=1)$. 
We assume the state entries $s_i$ are mutually independent Bernoulli random variables with known prior probabilities, given by a vector $\bomega=[\omega_1,\omega_2,\dots,\omega_N]$, where $\omega_i=P(s_i=0)$. 
{\color{black} A practical guideline to initialize the $\omega_i$'s is to set them to be uniform, and equal to a conservative estimate of the expected fraction of busy channels. In addition, past tests results could potentially be used to update the beliefs. 
Let us consider a reward $r_i>0$ for correctly detecting an empty sub-band and a penalty $\rho_i<0$ for failing to detect a busy sub-band, the utility can be written as 
\begin{align} 
&U(\bs,\N,K,\mathbf{B},\bdelta)\triangleq(K-\kappa)\!\sum_{i=1}^N\omega_ir_i\left(1-\alpha_i\right)+(1-\omega_i)\rho_i\beta_i\label{eq:utility_definition_empty_only}
\end{align} 
We anticipate however, that the framework proposed can be extended to cover applications where the utility comes from an action on the entries detected as busy, e.g. for a RADAR application, where the entries correspond to spatial directions.
Thus we use {\it SS case} (Spectrum Sensing) or {\it R case} (RADAR) to refer to the case where utility is generated by detection of empty or busy entries (also referred as resources), respectively. 
Extending the definition in \eqref{eq:utility_definition_empty_only} 
\begin{align} 
&U(\bs,\N,K,\mathbf{B},\bdelta)\label{eq:utility_definition}\\
&\triangleq\begin{cases}(K-\kappa)\!\sum_{i=1}^N\omega_ir_i\left(1-\alpha_i\right)+(1-\omega_i)\rho_i\beta_i &\mbox{{\color{black} SS case}}
\\(K-\kappa)\sum_{i=1}^N(1-\omega_i)(1-\beta_i)r_i+\omega_i\alpha_i\rho_i&\mbox{{\color{black} R case}} \nonumber
\end{cases}
\end{align} 
}
\begin{remark}
{\it In our framework there are two possible actions over a resource: the null action that always brings zero utility, and the other one that brings a random utility, which depends on the actual channel state.
This is motivated by the emphasis we place on the time-dependent utilization of the resources, which we assume occurs only in one of the binary states (based on the application of interest).
Note that a more general formulation for \eqref{eq:utility_definition} with $4$ different rewards/penalties (for the possible cases $(s_i,\delta_i)$) would not alter the structure of the problem, nor invalidate our results, i.e. there exists a unique mapping from our model to such case. 
}
\end{remark}
Finding the optimal policy corresponds to solve the following optimization problem 
\begin{equation}\label{eq:optimization_problem_intro}
\begin{aligned}
& \underset{\mathbf{B},\bdelta}{\text{maximize}}
& & \mathds{E}\left[U(\bs,\N,K,\mathbf{B},\bdelta)\right].\\
\end{aligned}
\end{equation}
\vspace{-.7cm}
\mbox{}
\nomenclature[v01]{$\A$}{Resources to sense}
\nomenclature[v02]{$\A_k$}{Resources to sense at time $k$}
\nomenclature[v03]{$\alpha, \upalpha$}{Type I error probabilities}
\nomenclature[v04]{$\beta, \upbeta$}{Type II error probabilities}
\nomenclature[v05]{$\mathbf{B}$}{Sensing matrix}
\nomenclature[v06]{$\bb_k$}{Mixing coefficients at time $k$}
\nomenclature[v07]{$\cB_i$}{Tests for resource $i$}
\nomenclature[v08]{$\mathcal{C}$}{Cycle for a single test}
\nomenclature[v09]{$\mathscr{C}$}{Set of cycles}
\nomenclature[v10]{$\bgamma$}{Test thresholds}
\nomenclature[v11]{$\bdelta$}{Decision rules}
\nomenclature[v12]{$\E$}{Set of edges}
\nomenclature[v13]{$\theta$}{Observations' distribution parameter}
\nomenclature[v14]{$K$}{Time horizon}
\nomenclature[v15]{$\kappa$}{Number of designed tests}
\nomenclature[v16]{$L$}{Maximum number of mixed sub-bands per test}
\nomenclature[v17]{$\N$}{Resources}
\nomenclature[v18]{$\boldsymbol{n}$}{Average received noise power}
\nomenclature[v19]{$\pi_{s_i}$}{Probability of declaring $\mathcal{H}_1$ given $s_i$}
\nomenclature[v20]{$r_i$}{Reward for resource $i$}
\nomenclature[v21]{$\rho_i$}{Penalty for resource $i$}
\nomenclature[v22]{$\bs$}{Resources' binary state}
\nomenclature[v23]{$\boldsymbol{\varphi}$}{Average received signal power}
\nomenclature[v24]{$y$}{Observation}
\nomenclature[v25]{$\bomega$}{Resources' prior belief}
\printnomenclature
\section{Dynamic Design of Sensing Matrices}\label{sec:design}
\subsection{Direct Inspection (DI) case}\label{subsec:DI}
In the DI case, we limit $\bb_k$ to have only one non-zero entry $i$, i.e. $b_{ki}\neq 0$, $b_{kj}=0~\forall j\neq i$. 
This means that there is an underlying hypothesis testing:
\begin{align*} 
\mathcal{H}_0&:y[k]\sim Exp~(\theta_0[k])\\
\mathcal{H}_1&:y[k]\sim Exp~(\theta[k])
\end{align*}
with $\theta_0[k]=b_{ki}n_i$ and $\theta[k]=b_{ki}(\varphi_i+n_i)>\theta_0[k]$.
In this context, it is known that the signal energy is a sufficient statistic for the test and that energy detection is optimal. 
Assuming no prior knowledge over the $\varphi_i$'s in case of existing communication, we only need to design the test threshold, which we set in order to maximize the utility defined in \eqref{eq:utility_definition}.
By defining $\theta^\star[k]\triangleq\max\{y[k],b_{ki}(\varphi_{\text{min}}+n_i)\}$ we get: 
\be\label{eq:energy_detection_test}
y[k]\overset{\mathcal{H}_1}{\underset{\mathcal{H}_0}{\gtrless}}\frac{\ln\left(\gamma_i\frac{\theta^\star[k]}{\theta_0}\right)}{\frac{1}{\theta_0}-\frac{1}{\theta^\star[k]}}
\ee 
\be\label{eq:threshold_gamma}
\gamma_i\triangleq\begin{cases}\frac{r_i\omega_i}{|\rho_i|(1-\omega_i)}&\mbox{{\color{black} SS case}}
\\\frac{|\rho_i|\omega_i}{r_i(1-\omega_i)}&\mbox{{\color{black} R case}}\end{cases}
\ee
Notice that, assuming a minimum average received signal power $\varphi_{\text{min}}>0$ in case of existing transmission, makes the test meaningful also for values of $\gamma_i<1$.
\begin{assumption}\label{ass:no_sensing_no_utility}
{\it 
To simplify the decision problem, we will assume every resource has to be sensed before being declared empty/busy.
This can be enforced as a standard/protocol rule or numerically guaranteed by setting
$\forall i\in\N, \omega_i<\frac{\rho_i}{\rho_i+r_i}$ ({\color{black} SS case}) / $\omega_i>\frac{r_i}{|\rho_i|+r_i}$({\color{black} R case}).
}
\end{assumption}
It is clear that the optimality \footnote{The threshold in \eqref{eq:threshold_gamma} is the optimal threshold that minimizes the Bayesian risk (maximize our utility) for the binary case, when $\varphi_i$ is known. It is of common practice to replace the MLE estimate for the unknown $\varphi_i$ (GLRT) and then reduce to the binary case, using the same threshold. A local most powerful test exists for $\theta\rightarrow\theta_0$ but GLRT is preferred for high SNR range.} of the test completely characterizes the set of decision rules $\bdelta$ for the sensed resources, while Assumption \ref{ass:no_sensing_no_utility} gives us the decision rules for the non-sensed resources. This implies that for the DI case, the optimization in \eqref{eq:optimization_problem_intro} can be expressed solely in terms of $\mathbf{B}$.
It is also known that for this type of test, where there is uncertainty in a parameter of the alternative hypothesis, one does not know the exact miss probability $\beta$; thus we will use an upper-bound, which will reflect in a lower bound for the achievable utility.
Since this test is part of the DI strategy, we add the superscript $^{DI}$ to the test error probabilities $\alpha_i$ and $\beta_i$ and have:
\begin{align}
\alpha^{DI}_{i}&=\min\left\{\left(\frac{|\rho_i|(1-\omega_i)}{r_i\omega_i\left(1+\frac{\varphi_{\text{min}}}{n_i}\right)}\right)^{\frac{1+\frac{\varphi_{\text{min}}}{n_i}}{\frac{\varphi_{\text{min}}}{n_i}}},1\right\}\\
\beta_{i}^{DI}&=1-\left(\alpha_{i}^{DI}\right)^{\frac{1}{1+\frac{\varphi_i}{n_i}}}.
\end{align}
What we can guarantee, since $\varphi_i\geq\varphi_{\text{min}}$ is that 
\be\label{eq:max_MD_prob} 
\beta_{i}^{DI}\leq 1-\left(\frac{\rho_i(1-\omega_i)}{r_i\omega_i\left(1+\frac{\varphi_{\text{min}}}{n_i}\right)}\right)^{\frac{n_i}{\varphi_{\text{min}}}}=\beta_{i,\text{max}}^{DI}
\ee 
\begin{remark}
{\it 
The test performance for the DI case does not depend on $b_{ki}$, therefore, for the DI case no further optimization is needed over the sensing matrix $\mathbf{B}$, other than selecting the non-zero entries.}
\end{remark}
Under Assumption \ref{ass:no_sensing_no_utility}, we can rewrite the optimization problem in \eqref{eq:optimization_problem_intro} for the DI case as
\begin{equation}\label{eq:optimization_problem_DI_case}
\begin{aligned}
& \underset{\A\subseteq\N}{\text{maximize}}
& & U^{DI}(\A)\\
\end{aligned}
\end{equation}
\begin{align}
U^{DI}\left(\A\right)&\triangleq(K-|\A|)\sum_{i\in\A}u_i^{DI}\\
u_i^{DI}&\triangleq \omega_ir_i(1-\alpha^{DI}_i)-(1-\omega_i)\rho_i\beta^{DI}_{i,\text{max}}\label{eq:U_i_def}
\end{align} 
We then introduce the following Lemma
\begin{lemma}\label{lemma:DI_submodular}
{\it
$U^{DI}({\A})$ is a normalized, non-monotone, non-negative sub-modular function
of $\A$.
} 
\end{lemma}
\begin{proof}
See Appendix \ref{app:proof_DI_submodular}
\end{proof}
Lemma \ref{lemma:DI_submodular} implies that there are diminishing returns in augmenting sets by adding a certain action to bigger and bigger sets.
The maximization of a non-monotonic sub-modular function is generally NP-hard, but the case of interest is not as difficult. In fact, by sorting the resources $i$
so that:
\be
u_1^{DI}\geq u_2^{DI}\geq \ldots\geq u_N^{DI}
\ee 
the set of size $i$, $\A_i=\{1,\ldots,i\}$ will be such that for any set ${\mathcal X}$
 of size $|{\mathcal X}|=i$
 $$\sum_{k=1}^i u_k^{DI}\geq \sum_{k\in {\mathcal X}}u_k^{DI}$$
Therefore, what remains is to find the best set size $i$ such that 
\be\label{eq:max_DI}
U^{DI}({\A})\leq U^{DI}(\A_i)\leq 
\max_{i}\left( (K-i)\sum_{k=1}^i u_k^{DI}\right)
\ee
The maximum in \eqref{eq:max_DI} is attained for 
\be\label{eq:stopping_rule_DI}
i^*=\inf\limits_{i}\{i: \partial_{i+1} U^{DI}({\A_i})<0\}
\ee 
where $\partial_{i+1} U^{DI}({\A_i})=(K-i)u_{i+1}^{DI}-\sum_{k=1}^{i+1}u_k^{DI}$.      
In fact, given the function is sub-modular, as soon as this condition is attained, it is maintained for $i+2,i+3$ etc., given that the marginal returns continue to decrease.
This maximization is {\it greedy} and stops when the marginal reward becomes negative. 
\subsection{A Group Testing Approach}\label{subsec:GT}
We now allow each test to mix different sub-bands, i.e. the vector $\mathbf{b}_k$ to have more than one non-zero entry. As outlined in the Introduction, aliasing of the spectrum comes with an associated noise folding phenomenon. Its impact is particular severe in a non coherent scheme as ours. In fact, the samples are collected sequentially and not in parallel, which means that we do not have multiple observations of the same value but only sequential observations tied to the same underlying random process.

To mitigate the noise folding effects, and reduce the hardware complexity, our focus is on low density measurement matrices. Our goal is to develop a relatively simple dynamic strategy for choosing a sensing matrix, whose utility can be expressed in closed form, and can potentially outperform the DI alternative.
A common approach for recovery with low density measurement matrices is to use belief propagation via message passing\footnote{In our model, an uninformative prior can be assigned to the $\varphi_i$'s to run the belief propagation message-passing algorithm on the obtained measurements}, whose most well known application is Low Density Parity Check (LDPC) optimum error correction decoding.
For LDPC (and CS methods), performance guarantees come as asymptotic bounds on the $\ell_2$-norm, but little is known for optimal design in the finite regime.
A difficulty in the design arises from the inherent multi-hypothesis testing problem associated with sensing several resources at the same time.
This is why, to develop our dynamic design, we look at a Group Testing (GT) approach, which allows us to consider a binary hypothesis test for each measurement.
In this way, the complexity of the analysis is relatively low, and we can derive the expected performance for any sensing matrix, under mild assumptions. 
 Prior to providing more details, a remark regarding related group-testing approaches is in order:
\begin{remark}{\it 
In the context of {\it group testing}, little is known in presence of measurement errors that depend on the group size, which is the scenario this work considers, as the remainder of the paper will detail. 
Asymptotic results on the target rate for measurement-dependent noise, using an information-theoretic approach, are given in \cite{kaspi2015searching}, where the noise is modeled as independent additive Bernoulli with bias dependent on the test size. Hence, the false-alarm and missed-detection probabilities of each single test, are symmetric. 
An additional noise, called {\it dilution effect} was considered in \cite{atia2012boolean}, where each resource could independently flip from $1$ to $0$ before the grouped test, and information-theoretic bounds were provided.   
In our model the false alarm and miss-detection probabilities are dependent on the optimization of the test threshold, therefore the noise is not independently added (nor an independent dilution can be considered).
Furthermore, the strategy derived depends on the finite horizon for $K$, i.e. our results are not asymptotic. The same considerations apply to similar information-theoretic approaches in \cite{scarlett2016converse,chan2011non,sharma2015finding}.}
\end{remark}

From the sensing matrix $\mathbf{B}$, let us define the sets $\A_k=\{i\in\N:b_{ki}\neq0\}$ and $\mathcal{B}_i=\{1\leq k \leq\kappa:b_{ki}\neq 0\}$.
Note that at times we use $\mathbf{B}$  as an argument in functions that, strictly speaking, are just functions of the sets $\A_k$ just defined.
For each test we define a binary {\it group test} as follows\footnote{\color{black}
We envision that such test would be useful for a downlink transmission in which the Access Point (AP) may want to allow multiple communications at the same time and can alert the SUs over a narrowband signaling channel to access the spectrum.}: 
\be\label{eq:group_hyp_test}
\begin{cases}\mathcal{H}_0:&\forall i\in\A_k~~s_i=0\\
&\Rightarrow\theta_0[k]=\bb_k^T\boldsymbol{n}\\
\mathcal{H}_1:&\exists i\in\A_k ~~s.t.~~ s_i=1\\
&\Rightarrow\theta[k]\geq\!\left(\!\min\limits_{i}b_{ki}\!\right)\!\varphi_{\text{min}}\!+\!\bb_k^T\boldsymbol{n}=\theta_{\text{min}}[k]
\end{cases}
\ee
\begin{remark}
{\color{black}\it  It is important to highlight that the two hypotheses pertain exclusively the group of sub-bands explored in test (i.e. $\A_k$), not the whole spectrum.
Also note that this group-testing approach pertains the design of the sensing matrix and detection algorithm and not the underlying observation model. The different approaches we compare ourselves against later, use detection strategies that are multi-hypothesis tests.}  
\end{remark}

The test can be written  as:
\be 
\frac{\max\limits_{\theta[k]\geq\theta_{\text{min}}[k]}f_{\theta[k]}(y[k])}{f_{\theta_0[k]}\left(y[k]\right)}\overset{\mathcal{H}_1}{\underset{\mathcal{H}_0}{\gtrless}}\gamma_k.
\ee for which we can derive: 
\begin{align} 
\upalpha(\bb_k,\gamma_k)&=\left(\frac{1}{\gamma_k\frac{\theta_{\text{min}}}{\theta_0}}\right)^{\frac{\frac{\theta_{\text{min}}}{\theta_0}}{\frac{\theta_{\text{min}}}{\theta_0}-1}}\\
\upbeta(\bb_k,\gamma_k)&=1-\left(\upalpha(\bb_k,\gamma_k)\right)^{\frac{\theta_0}{\theta_{\text{min}}}}
\end{align}
The decision declares that resource $i$ is busy ($\mathcal{H}_1$ is true) if the majority of the tests, where resource $i$ is involved, is positive, else it accepts the null hypothesis $\mathcal{H}_0$ for resource $i$.
Thus: 
\begin{align}
&\pi_0(i,\bb,\gamma)=\!\!\left(\!1\!-\!\!\!\!\!\prod_{j\in\A_k\setminus i}\!\!\!\!\omega_j\!\right)\!\left(1\!-\!\upbeta_{i}(\bb,\gamma;0)\right)+\!\upalpha(\bb,\gamma)\!\!\!\!\!\prod\limits_{j\in\A_k\setminus i}\!\!\!\!\omega_j\label{eq:pi_0_def}\\
&\pi_1(i,\bb,\gamma)=1-\upbeta_{i}(\bb,\gamma;1)\label{eq:pi_1_def}
\end{align}
where the functions $\pi_j(i,\bb,\gamma),~j=0,1$ are only defined when $b_i\neq 0$. These functions represent the probabilities of declaring {\color{black} $\mathcal{H}_{1}$} in a group-test defined by $\bb$ with threshold $\gamma$ and given $s_i=j,~j=0,1$.
Notice that the error probabilities $\upalpha,\upbeta$ refer to each binary hypothesis testing defined in \eqref{eq:group_hyp_test}. 
The notation for $\upbeta_{i}(\bb,\gamma;s_i)$ indicates the probability of having a missed-detection conditioned on the state $s_i$ of one of the resources.
It then follows that 
\begin{align}
&\alpha^{GT}_i(\mathbf{B},\boldsymbol{\gamma})\label{eq:alpha_GT}\\
&\triangleq\!1\!-\!F_{PBD}\!\left(\left\lceil\frac{|\mathcal{B}_i|}{2}\right\rceil-1;|\mathcal{B}_i|,\{\pi_0(i,\bb_k,\gamma_k)\!:\!k\in\mathcal{B}_i\}\right)\nonumber\\
&\beta^{GT}_i(\mathbf{B},\boldsymbol{\gamma})\label{eq:beta_GT}\\
&\triangleq F_{PBD}\!\left(\left\lceil\frac{|\mathcal{B}_i|}{2}\right\rceil-1;|\mathcal{B}_i|,\{\pi_1(i,\bb_k,\gamma_k)\!:\!k\in\mathcal{B}_i\}\right)\nonumber
\end{align} 
where $F_{PBD}(k;n,\mathbf{p})$ indicates the CDF of a Poisson Binomial Distribution parameterized by $\mathbf{p}\in[0,1]^n$.
At this point, one can replace \eqref{eq:alpha_GT}-\eqref{eq:beta_GT} in \eqref{eq:utility_definition}, to then solve the optimization in \eqref{eq:optimization_problem_intro}, where the equivalence between the decision rules $\bdelta$ and the selection of the thresholds $\bgamma$ is essentially the same as for the DI case.

{\color{black} 
Notice that, in order for \eqref{eq:alpha_GT}-\eqref{eq:beta_GT} to hold, each of the tests must  be independent, conditioned on the state of the resource $i$. This is true if the sensing matrix (in the language used for LDPC codes) does not have length-4 cycles (i.e. two different measurements do not mix more than one sub-band in common)\footnote{Such condition is typically required for belief propagation algorithms, e.g. message passing, which suffer from loopy networks with short cycles.}. }

The optimization remains extremely complex due to the complexity of the decision space for $\mathbf{B}$ and the sum of an exponentially growing number of terms for the probabilities defined in \eqref{eq:alpha_GT}-\eqref{eq:beta_GT}. Nevertheless, it gives a method to evaluate the objective of our optimization for any sensing matrix $\mathbf{B}$, where the optimization over $\bgamma$ can be numerically solved.
In fact,  \eqref{eq:alpha_GT}-\eqref{eq:beta_GT} are monotonic functions of the probabilities $\pi_0,\pi_1$ defined in \eqref{eq:pi_0_def}-\eqref{eq:pi_1_def}, which are monotonic in the $\gamma_k$'s, and therefore a unique solution for $\bgamma$ exists.
Next, we introduce additional constraints to \eqref{eq:optimization_problem_intro}, in particular on the structure of $\mathbf{B}$, in order to evaluate whether a GT strategy could be superior to the DI approach.
\begin{remark}
{\it Note that an ML or a MAP estimator, for a rank-deficient sensing matrix, do not provide optimality guarantees in terms of minimum error probability or minimum Bayesian risk. Nevertheless, for the same sensing matrix, we expect the MAP estimator to outperform the binary {\it group-testing} hypothesis in \eqref{eq:group_hyp_test} by simply adding more degrees of freedom to the decision $\bdelta$ in the $\kappa$-th dimensional space of the observations. Therefore, the evaluation of the objective in \eqref{eq:optimization_problem_intro} via \eqref{eq:alpha_GT}-\eqref{eq:beta_GT} provides a benchmark for the utility obtainable with a more refined detection method.}
\end{remark}

\subsubsection{The pairwise tests case}\label{subsubsec:L_2}
We start by considering matrices $\mathbf{B}$ that have the following property: each resource is sensed only one time, either directly inspected or mixed with another resource, and no test mixes more than 2 resources, i.e. $|\A_k|\leq L=2,|\mathcal{B}_i|\leq 1~~\forall k=1,\dots,\kappa,~i=1,\dots,N$.
Let us discuss the test that mixes entries $i$ and $j$. According to the strategy derived at the beginning of the section, one can use \eqref{eq:pi_0_def}-\eqref{eq:pi_1_def}-\eqref{eq:alpha_GT}-\eqref{eq:beta_GT} to write out the per-time instant utility obtainable after the decision. 
First, from \eqref{eq:group_hyp_test}, we note that, without prior knowledge over $\varphi_i,\varphi_j$ other than the threshold $\varphi_{min}$, the best choice to minimize $\alpha$ is to set $b_{i}=b_{j}$ (we refer to this false alarm probability as $\alpha_{ij}$). 
Therefore, similar to the DI case, one can consider binary coefficients for $\bb_k$, i.e. $b_{ki}\neq 0\rightarrow b_{ki}=1$. This will hold true also for the extension of $L>2$ and will give implementation advantages as discussed in \ref{subsec:hard_limit}.

A missed detection event in \eqref{eq:group_hyp_test} can occur for three different states of the resources $i,j$; we upper-bound the corresponding missed detection probabilities by always considering $\theta=\theta_{\text{min}}$ and refer to this bound as $\beta_{ij,\text{max}}$. 
We then obtain:
\begin{align}
&u_{ij}^{GT}\triangleq \omega_i\omega_j(r_i+r_j)(1-\alpha_{ij})+\left[(\omega_i(1-\omega_j)(r_i+\rho_j)+\right.\nonumber\\
&\left.\omega_j(1-\omega_i)(r_j+\rho_i)+(1-\omega_i)(1-\omega_j)(\rho_i+\rho_j))\right]\beta_{ij,\text{max}}\label{eq:u_ij_GT_def}
\end{align}
where the threshold for this test $\gamma_{ij}$ has been set to maximize \eqref{eq:u_ij_GT_def}, i.e. 
\be\label{eq:gamma_L_2}
\gamma_{ij}=\frac{\omega_i\omega_j(r_i+r_j)}{(1-\omega_i)(|\rho_i|-\omega_jr_j)+(1-\omega_j)(|\rho_j|-\omega_ir_i)}.
\ee 
Let us then consider a graph where each resource is a vertex and the edge weight $u_{ij}$ between two vertices $ij$ is the utility (per time instant) $u^{GT}_{ij}$ just defined (the weight of the loops $u_{ii}^{GT}$ are given by $u^{DI}_{i}$ in \eqref{eq:U_i_def}). 
We can then translate our problem into a particular instance of a {\it max-cut problem}: picking a subset of the edges and form a subgraph, where each edge represents a test, to maximize the objective in \eqref{eq:optimization_problem_intro}.
Formally, we can write
\begin{equation}\label{eq:submodular_problem_L_2_R_1_initial}
\begin{aligned}
& \underset{\mathcal{E}}{\text{maximize}}
& & U^{GT}(\E)\\
&\text{subject to}  & & \deg_{\E}(i)\leq 1~~\forall i\in\N
\end{aligned}
\end{equation}
where 
\be\label{eq:tot_utility_def}
U^{GT}(\E)\triangleq\!\left(K\!-|\mathcal{E}|\right)\!\!\left(\sum_{ij\in\mathcal{E}}u_{ij}^{GT}\right)\\
\ee 
and $\deg_{\mathcal{E}}(i)$ is the nodal degree of node $i$ induced by the undirected graph $\mathcal{G}=\left(\N,\E\right)$.
It is possible to map the constraint on the nodal degree in the objective of \eqref{eq:submodular_problem_L_2_R_1_initial}, by adding a penalty for the violation of such constraint. 
This guarantees the optimal solution will be equivalent to \eqref{eq:submodular_problem_L_2_R_1_initial}, i.e. no set of edges that violates the constraint can improve the objective, and any feasible set of edges would have the same objective in the two problems.
We rewrite our optimization as
\begin{equation}\label{eq:submodular_problem_L_2_R_1}
\begin{aligned}
& \underset{\mathcal{E}}{\text{maximize}}
& & U^{GT}(\E)-M\sum_{i\in\N}\Upsilon(\deg_{\E}(i))\\
\end{aligned}
\end{equation}
where 
\be 
\Upsilon(n)\triangleq\begin{cases}0&\mbox{for}~n\leq 1\\
n-1&\mbox{for}~n\geq 2\end{cases}
\ee 
and $M$ is a positive constant. 
\begin{lemma}\label{lemma_submodular_mixed}
{\it For $M>0$ the objective in \eqref{eq:submodular_problem_L_2_R_1} is a non-monotone sub-modular function of $\E$ and it is possible to find $M^*>0$ such that for any $M>M^*$ the two optimizations \eqref{eq:submodular_problem_L_2_R_1_initial}-\eqref{eq:submodular_problem_L_2_R_1} are equivalent.}
\end{lemma}
\begin{proof}
See Appendix \ref{app:proof_submodular_mixed}
\end{proof}
We now discuss the extension of this result for $L>2$, to develop a general algorithm that leverages the sub-modularity of the optimum design problem in \eqref{eq:submodular_problem_L_2_R_1_initial}.
\subsubsection{Extension to $L>2$}\label{subsubsec:extension_L_greater_2}
If we mix more than 2 channels, instead of just edges or self loops to indicate the tests, we could have cycles of length up to $L$. 
The nodal degree in \eqref{eq:submodular_problem_L_2_R_1} will then be interpreted as the number of cycles a node belongs to, and the set of edges will be replaced with the set of cycles. 
We then replace the set $\E$ of edges with the set $\mathcal{C}$ of possible cycles, and use $c$ to indicate the generic cycle (which could be a self-loop, an edge or a cycle with length $3$ or greater).
With these substitutions, the proof of sub-modularity in Lemma \ref{lemma_submodular_mixed} naturally extends to this case as well.
In light of the constraint $|\mathcal{B}_i|\leq 1$ we will have that no node can be in two cycles. 

To visualize this concept, in Fig.\ref{fig:cycles_condition} we show two possible sets of cycles of length up to $4$. On the right, we have a set of tests that respect our constraint: there is a test that only considers one resource and three tests that combine 2, 3, and 4 resources respectively, but no resource is considered in two different tests. 
On the left, instead, a resource is considered in two tests: one where it is combined with other $3$ resources, and one where it is inspected directly; such configuration is therefore not acceptable.
\begin{figure}[ht]
\centering
\includegraphics[width=\linewidth]{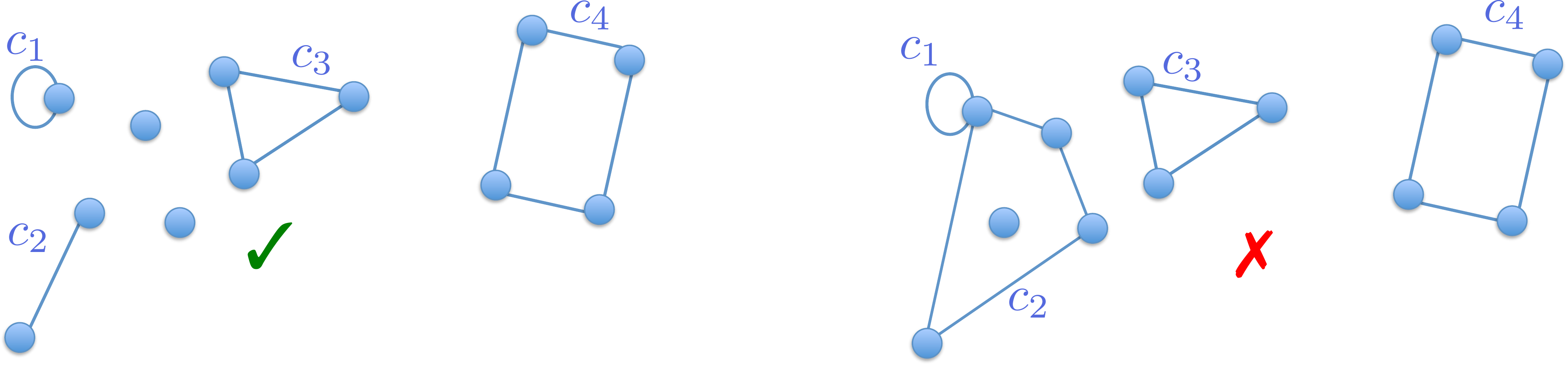}
\caption{Example of two sets of tests. The right configuration has $4$ tests and no resource is considered in two different tests, therefore it respect our constraints, whereas the left configuration has a resource included in two tests, and is not a feasible solution.}\label{fig:cycles_condition}
\end{figure}

\subsubsection{The factor approximation of the greedy algorithm}
Having proven the sub-modularity of \eqref{eq:submodular_problem_L_2_R_1} in Lemma \ref{lemma_submodular_mixed}, it is natural to resort to a greedy procedure, however it is important to highlight that the objective in \eqref{eq:submodular_problem_L_2_R_1} does not respect the non-negativity property.
To the best of our knowledge, there is no known procedure in the literature on approaching the maximization of a general sub-modular non monotone function, if the minimum value is not known: no constant approximation factor guarantee can therefore be given in general.
Nevertheless, due to the particular structure of our problem, it is possible to find a factor approximation for the output of the greedy procedure, described in Algorithm \ref{alg:mixed_maximization_L_2}. 
\begin{lemma}\label{lemma_approximation_factor}
{\it Algorithm \ref{alg:mixed_maximization_L_2} guarantees a $\alpha$-constant factor approximation of the optimal solution for \eqref{eq:submodular_problem_L_2_R_1}, where:
\begin{equation}
\alpha=\frac{1}{\min\{L_{\text{eff}},\frac{K}{2}\}}\frac{K-1}{K-\min\{L_{\text{eff}},\frac{K}{2}\}}.
\end{equation}
}
\end{lemma}
\begin{proof}
See Appendix \ref{app:proof_approximation_factor}.
\end{proof}
\begin{algorithm}[!t]
\caption{Greedy Maximization of $U^{GT}(\mathscr{C})$}
\begin{algorithmic}[1] \label{alg:mixed_maximization_L_2}
\STATE {\bf Initialize}: $\mathscr{C}=\emptyset$. 
\STATE {\bf While} $\exists~\C\in\bar{\mathscr{C}}$ such that $\partial_{\C}U^{GT}(\mathscr{C})>0$
\STATE  ~~~~~~Find $\C^*=\argmax_{\C\in\bar{\mathscr{C}}}\partial_{\C}U^{GT}(\mathscr{C})$
\STATE ~~~~~~$\mathscr{C}\leftarrow\mathscr{C} \cup \mathcal{C}^*$
\STATE {\bf End}
\end{algorithmic}\end{algorithm}

The set $\bar{\mathscr{C}}$ indicates the set of cycles that are not adjacent (share a node) with any of the cycle in $\mathscr{C}$.
Note that
\be\label{eq:finite_diff}
\partial_{\C'}U^{GT}(\mathscr{C})=-\sum_{\C\in\mathscr{C}}u_{\C}+(K-|\mathscr{C}|-1)u_{\C'}
\ee 
so, as long as the number of tests $|\mathscr{C}|$, added in the greedy maximization, is less than the time horizon $K$, we have
\be  
\argmax_{\C\in\mathscr{C}^c}\partial_{\C}U^{GT}(\mathscr{C})=\argmax_{\C\in\mathscr{C}^c}u_{\C}.
\ee 
This relation indicates that, in the greedy procedure, edges are added in decreasing order of utility, respecting the constraint on the nodal degree. Also, from \eqref{eq:finite_diff} it is easy to find that the optimal $|\mathscr{C}|$ will never exceed $\left\lceil\frac{K-1}{2}\right\rceil$.
In the greedy procedure in Algorithm \ref{alg:mixed_maximization_L_2}, there is a constant number of operations per query, which indicates the overall complexity of the algorithm is dominated by the sorting of all possible cycles' utilities. 
In the worst case, sorting $n$ values require $O(n^2)$ operations, thus the complexity will be given by $O\left(\left(\sum_{\ell=1}^{L}\binom{N}{\ell}\right)^2\right)=O\left(N^{2L}\right)$, i.e. polynomial in $N$ and exponential in $L$.
\subsection{Additional applications of the stochastic optimization}
We would like to highlight the analysis for the factor approximation of the greedy strategy transcends the spectrum sensing application discussed in detail in this paper. In fact, group-testing has been applied to a number of disparate contexts to model the outcome of sequential tests.
As long as one has a way to define the per-time utility derived from each test as in \eqref{eq:u_ij_GT_def}, and an overall utility as in \eqref{eq:tot_utility_def}, then our results can be applied.  Classes of problems that could be formulated in a similar way include job scheduling for data centers, design of parity checks for rateless coding, dynamic advertisement (promoting an offer that bundles two products/services together) etc.. Obviously, in all these cases, the statistics of the observations would be radically different. 
\section{Alternative approaches}\label{sec:additional}
\label{sec:CS_ML_estimate}
In the previous sections, we have provided methods that find a low density measurement matrix. As will be apparent in our numerical results, the noise folding phenomenon justifies the use of sparse sensing matrices. They are also ideal when one wants to use belief propagation to the decision problem. However, for the sake of comparison here we look at alternative support recovery methods, which can be applied to any measurement matrix $\mathbf{B}$, and that can be mapped into previous solutions, as the MWC in \cite{mishali2010theory,cohen2014sub}.  
\subsection{ML estimate}
Let us assume that $\kappa$ measurements have been collected, by mixing a set $\A\subseteq \N$ of sub-bands. One could ignore the prior $\omega_i$ and derive the ML estimate for $\varphi$. 
The log-likelihood function is:
\begin{align}
&\log\left(f\left(\boldsymbol{y}|\boldsymbol{\varphi}_{\A}\right)\right)=-\sum_{k=1}^{\kappa}\log\theta[k]+\frac{y[k]}{\theta[k]}\nonumber\\
&\overset{\theta[k] \rightarrow y[k]}{\approx}-\sum_{k=1}^{\kappa} 1+\log y[k]+\frac{1}{2}\left(\frac{y[k]-\theta[k]}{y[k]}\right)^2\label{eq:ML_approx}
\end{align} 
where the linearization corresponds to the Taylor expansion of the likelihood function around the observations mean (recall \eqref{eq:theta_def}-\eqref{eq:y_def}).
A possible approach consists in solving the following LASSO problem:
\be\label{eq:LASSO}
\hat{\boldsymbol{\varphi}}_{\A}=\argmin_{\boldsymbol{\varphi}_{\A}}\|\boldsymbol{\lambda}_{\A}\boldsymbol{\varphi}^T_{\A}\|_1+\frac{1}{2}\|\left(\mathbf{y}-\mathbf{B}(\boldsymbol{\varphi}^T_{\A}+\boldsymbol{n}_{\A}^T)\right)\|_{\mathbf{C}^{-1}}^2
\ee 
with $\mathbf{C}=\text{diag}(\mathbf{y})$ denoting the covariance of the observations, and $\boldsymbol{\lambda}_{\A}$ the vector of weights for the weighted $\ell_1$-norm. The first penalty term in the objective enhances sparsity, while the second term comes from the ML estimate in \eqref{eq:ML_approx}. 
To incorporate the information of the prior beliefs $\omega_i$, one can set $\lambda_i=\gamma_i$ from \eqref{eq:threshold_gamma}, $\forall i\in\A$, to favor the estimates $\varphi_i>0$ for entries with lower thresholds $\gamma_i$. Alternatively, one can also set $\lambda_i=\lambda~\forall i\in\A$. 
Note that, compared to the non-sequential sampling models (i.e. those using a filterbank), the application of the LASSO (see Section \ref{sec:CS_ML_estimate}) in this context is an approximation. 
The random demodulator in \cite{tropp2010beyond} (similar to our scheme in terms of architecture) is an Xampling ADC converter for signals that are sum of harmonics with constant amplitude, i.e. each subspace, in the UoS representation, has finite dimension. {\color{black} This is not the case we are interested in, and approximating our signal as a sum of harmonics would require sampling at a much higher rate than $R_s$.} 
For our {\it multiband} signal model, instead, rather than having observations that are noisy linear combination of a sparse input, the p.d.f. of the samples depends on those same linear combinations.
{\color{black}
\subsection{Covariance estimate}\label{sec:MWC}
A similar approach is to estimate the covariance of the samples, and write the correspondent linear equations system, as derived in \cite{cohen2014sub}, using the analog front-end of the MWC, introduced in \cite{mishali2010theory}.
This leads to write a system 
\begin{equation}\label{eq:linear_system_MWC}
\boldsymbol{z}=\mathbf{B}\left(\boldsymbol{\varphi}+\boldsymbol{n}\right)+\boldsymbol{\epsilon}
\end{equation} 
where the sensing matrix $\mathbf{B}$ ($\mathbf{A}$ in their work) is a $M\times N$ matrix, with $M$ being the number of analog channels and $N$ the number of spectral bands, whose occupancy is desired to be detected.
The authors estimate the covariance vector (diagonal of the covariance matrix) $z_i=\mathds{E}\left[x_i^2[k]\right]$ (where the $x_i$ have been introduced in \eqref{eq.model-expanded}) by taking multiple samples in multiple frames (in their work $K$ samples per frame in $P$ frames).
Note, however, than in their work the sampling frequency per branch, called $f_s$, needs to be larger than the single component bandwidth $R_s$, to justify the approximation of a multi-band signals as sum of harmonics with constant amplitudes over a single frame. Also, in light of this, the different frames considered for the estimate of $\boldsymbol{z}$, cannot be consecutive in time, since the $x_i$'s would be correlated. 
However, for the sake of comparison, to calculate the utility of the scheme in \cite{cohen2014sub}, we will ignore this limitation in our numerical tests (c.f. Section \ref{sec:simulations}) as well as their need of sampling at a faster rate than our method, pretending their scheme can take $\kappa$ (using their notation $\kappa=KP$) independent consecutive measurements, and can do so at the lower rate $R_s$. One can then use \eqref{eq:LASSO} for the recovery of the sparse vector $\boldsymbol{\varphi}$ from the linear system in \eqref{eq:linear_system_MWC}, replacing $\mathbf{C}$ with $\mathbf{I}$, and $\A$ with $\N$, since their scheme mixes the whole spectrum. 
}
\section{Simulation Results}\label{sec:simulations}
In this section, we showcase the ability of our approach to dynamically switch between a DI receiver (scanning receiver) and a GT approach, based on the expected occupancy (the vector of priors $\bomega$), the time available $K$, the minimum SNR threshold $SNR_{\text{min}}=\frac{\varphi_{\text{min}}}{w}$ and the number of resources $N$.
In the context of spectrum sensing (SS case), the parameters $r_i$ and $\rho_i$ can be mapped into a maximization of the overall weighted network throughput (see \cite{ferrari2017utility}): the reward $r_i$ can be proportional to the achievable rate over the channel $i$ in the absence of PU communications, i.e. $r_i\propto\log(1+SNR_{i,S})$  (where the suffix $S$ indicates the secondary communication), whereas the penalty $\rho_i$ can be made proportional to the loss in rate caused to the primary communication, due to the interference added by the secondary. 
For the cognitive radio application, the concept of {\it exploitation} of the resource is tied to the definition of utility function chosen in \eqref{eq:utility_definition}, which is expressed in $bits/s/Hz$\footnote{From \eqref{eq:utility_definition}-\eqref{eq:optimization_problem_intro}, $r_i$'s and $\rho_i$'s can be normalized over the communication bandwidth without altering the optimization.}.
The longer the time available to transmit, the larger is the number of bits that can be transmitted over that band. 
For the other case, i.e. when the reward comes from detecting correctly resources that are busy (for example a RADAR application), it is not immediately clear why the utility would be proportional to the number of remaining time instants.
To interpret this, we model the action upon declaration of $s_i=1$ as a Bernoulli trial which accrues a reward $r_i$ if such action is successful (i.e. the target is actually hit) and this happens with a certain probability $p_i$ for each attempt. 
The number of attempts $T_i$ necessary to hit the target will then be geometrically distributed. 
One can then find that the expected reward is equal to $r_iP(T_i\leq(K-\kappa))=r_i\sum_{k=1}^{K-\kappa}p_i(1-p_i)^{k-1}=r_i(1-(1-p_i)^{K-\kappa})\approx (K-\kappa)r_ip_i$ for small $p_i$, which would motivate having an expected utility that increases linearly with time. 
The $\rho_i$ associated with this case would model an {\it intervention cost}, whose main purpose would be to limit the false alarm rate. 
It is important to highlight, however, that the time dependency in the optimization objective prevents our formulation to return a standard Constant False Alarm Rate (CFAR) detection method. 
Nevertheless, our model can apply to electronic warfare (tentatives of create jamming), wake-up radio, and other problems where the action (and the associated utility) is on the channels that are declared busy. 
For all the figures we refer to $L=2,3$ as the maximum number of resources per test allowed in our greedy procedure in Algorithm \ref{alg:mixed_maximization_L_2}. 
Theoretically, the optimal value for $U^{GT}$ monotonically increases with $L$, since increasing $L$ introduces additional degrees of freedom. 
However, in our simulations we used the greedy solution and, as proved in our Lemma \ref{lemma_approximation_factor}, the approximation factor of the greedy maximization is potentially worse for higher values of $L$, as the following numerical results will show.
 
We indicate with ``Group Testing" the utility obtained with our GT approach. 
The ``MAP Estimator" is the estimator that knows the true values $\varphi_i$, uses the same matrix $\mathbf{B}$ of the GT approach, but then decides on each resource, based on the posterior for $\omega_i$, using belief-propagation. 
\subsubsection{SS case vs RADAR}
Even though, in light of the symmetry in the definition of the threshold $\gamma_i$, one can switch the $r_i$'s and $\rho_i$'s to go from {\color{black} SS case} to {\color{black} R case} and find the same trends, even for the combined tests, to avoid confusion, we highlight the difference in the two scenarios, in the first simulation we present in Fig.\ref{fig:SS_RADAR}.

For the experiment in Fig. \ref{fig:SS_RADAR}, we set $K=30,N=60$ and $r_i=r, \rho_i=\rho ~\text{and} ~\omega_i=\omega, SNR_i=SNR_{\text{min}} (10 dB)~\forall i\in\N$ we have that for  $\omega$ equal to $\frac{\rho}{\rho+r}$ or $\frac{r}{\rho+r}$ for SS case or R case, respectively.
These are the threshold values given in Assumption \ref{ass:no_sensing_no_utility}, to guarantee no resource can give positive utility if not tested.
As we can see, in both scenarios the utility increases with the ratio $\frac{\rho}{r}$, since the prior probability, that favors a positive utility, increases as well. However, the gain for the GT approach over the DI, happens in complementary ranges: when $\frac{\rho}{r}>1$ for spectrum sensing application, and when $\frac{\rho}{r}<1$ in the RADAR problem.
When the penalty increases with respect to the reward, the GT approach for spectrum sensing will be conservative and not transmit in any of the channels in a {\it group} that tested positively;
nevertheless, as the priors $\omega_i$ increase, it is possible to find multiple empty sub-bands with just one test and gain in utility compared to the DI.
For the RADAR application, when the penalty increases with respect to the reward, there is a disadvantage in declaring as busy all the elements in the test, even if the prior $\omega$ decreases. Clearly this limits the benefit of combined tests, whereas when $\frac{\rho}{r}$ decreases, there is a gain since one element, found busy in the pool, guarantees higher reward.
Apart from this asymmetry, both cases show the same trends in utility over number of available resources $N$, and the value of $SNR_{\text{min}}$. Hence, we will only consider the SS case in the next simulations. 
\begin{figure}[t]
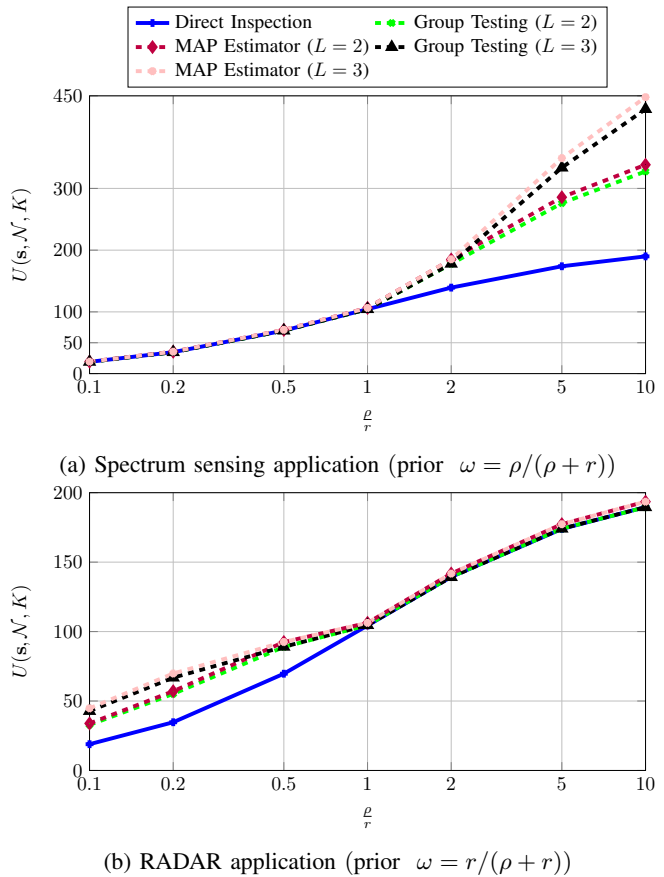

\centering
\begin{subfigure}{\linewidth}
\centering
\includegraphics[width=\linewidth]{Figures/K_30_varrho_vec_empty.tikz}
\caption{Spectrum sensing application $\left(\text{prior}~~\omega={\rho}/({\rho+r})\right)$}
\end{subfigure}
\begin{subfigure}{\linewidth}
\centering
\includegraphics[width=\linewidth]{Figures/K_30_varrho_vec_full.tikz}
\caption{RADAR application $\left(\text{prior}~~\omega={r}/({\rho+r})\right)$}
\end{subfigure}
\caption{Comparing utility for different approaches vs the ratio $\frac{\rho}{r}$ (horizon $K=30$, number of resources $N=60$, $SNR_{\text{min}}=10 dB$). The utility for the RADAR application ({\color{black} R case}) is normalized over the unit measure of $r_i$ and $\rho_i$.}\label{fig:SS_RADAR}
\end{figure}

\subsubsection{Utility for different $N$}
In Fig.\ref{fig:N_vec} we plot the utility (normalized over $K^2$) over the ratio $\frac{K}{N}$ for two different horizons, i.e. $K=10$ and $K=30$ and $SNR_{\text{min}}=10dB$.
 We can see that, only for $\frac{K}{N}\lessapprox 0.75$, the GT approach outperforms all competing options whereas, when the horizon increases, almost no benefit comes from mixing resources. This suggests that there is enough time to test them independently with high accuracy.
For this experiment, we set $\omega_i\sim\left(0.7,\frac{\rho_i}{\rho_i+r_i}\right)$, where $r_i=\log(1+SNR_{i,S})$ and $\rho_i=5r_i$ with $SNR_{{i,S}_{dB}}\sim\mathcal{U}([10,20])$. The $SNR$ for the test,  $\frac{\varphi_i}{n_i}$, is generated uniformly between $10$ and $20$ dB; recall that the only information used in our algorithm is the minimum $SNR$ value, i.e. in this case $10 dB$.
In the regime considered, the DI is approximately constant since it is easy to show $U^{DI,OPT}\leq\frac{K^2}{4}u_{\text{max}}$ irrespective of $N$. 
\subsubsection{Utility for different $SNR_{\text{min}}$}
In this set of experiments we studied how the utility behaves versus the minimum $SNR$ in each active sub-band. In this case the $SNR$ was drawn uniformly between $SNR_{\text{min}_{dB}}$ and $SNR_{\text{min}_{dB}}+10$, and once again only the value of $SNR_{\text{min}}$ was used in the optimization, which is shown in the abscissa of the figures.
Matching our intuition, we can see how the GT approach outperforms the DI only when $SNR_{\text{min}}$ is sufficiently high, and also that the gain in utility is larger for $K=10$ than for $K=30$. In fact, for this experiment the number of resources has been fixed to $N=20$ and, as previously highlighted, increasing $K$ for fixed $N$ diminishes the benefit of combining resources in a test. 
We also plotted the utility obtainable with the approximate ML estimate obtained via Compressive Sensing, described in Section \ref{sec:CS_ML_estimate}. For this case, to illustrate the noise-folding issue, we used a dense matrix that has the same aspect ratio of the one  found via $GT$ approach (i.e. that scans the same set of resources for the same number of tests). 
To show reasonable results, only for the ML estimate via CS, we actually took the mean of $y[k]$ over $10$ samples.
We can see that, despite having more measurements, such approach gives a much lower utility than DI as well as the proposed GT, due to the negative effect of noise folding. 
For $K=30$, we also compared our approach with the performance obtained using belief propagation in a loopy network and an LDPC matrix (see \cite{baron2010bayesian} for details). 
Considering $N=20$ resources, and an expected sparsity approximately equal to $4$, we chose a regular LDPC matrix with a row weight of $5$ ($20/4$ as suggested in \cite{baron2010bayesian}) and a column weight of $3$, resulting in $12$ tests. 
The LDPC has not been implemented for $K=10$, since the regularity constraints would have given either a diagonal matrix (same as DI), or a relatively dense matrix. 
The absence of any optimization in the choice of which and how many resources to test produces a utility which, for low $SNR$, is lower than the DI approach proposed.
For high enough $SNR$, the LDPC design can outperform the DI approach, but still gives a utility lower than our GT strategy with $L=2$.
This highlights the benefit of having an active sub-Nyquist receiver compared to a static offline selection of the parameters. 
\begin{figure}[t]
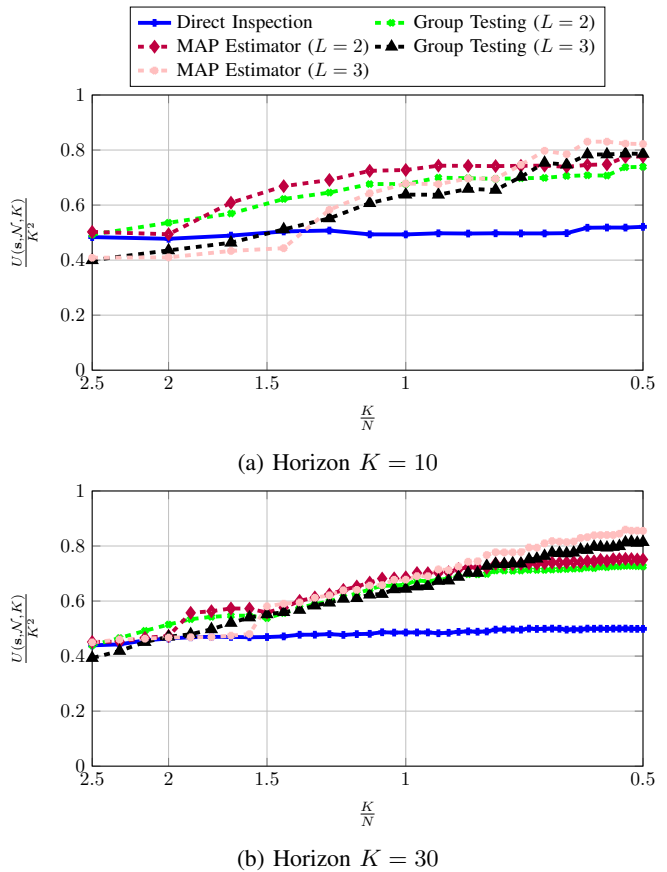

\centering
\begin{subfigure}{\linewidth}
\centering
\includegraphics[width=\linewidth]{Figures/K_10_N_vec_empty.tikz}
\caption{Horizon $K=10$}
\end{subfigure}
\begin{subfigure}{\linewidth}
\centering
\includegraphics[width=\linewidth]{Figures/K_30_N_vec_empty.tikz}
\caption{Horizon $K=30$}
\end{subfigure}
\caption{Comparing utility for different approaches vs the ratio horizon $K$ over the number of resources $N$ for different horizons $K$ and the prior $\omega_i\sim\mathcal{U}\left(0.7,\frac{\rho_i}{r_i+\rho_i}\right)$. The utility on the $y$ axis is normalized by $K^2$ ($SNR_{\text{min}}=10dB$)}\label{fig:N_vec}
\end{figure}
\begin{figure}[t]
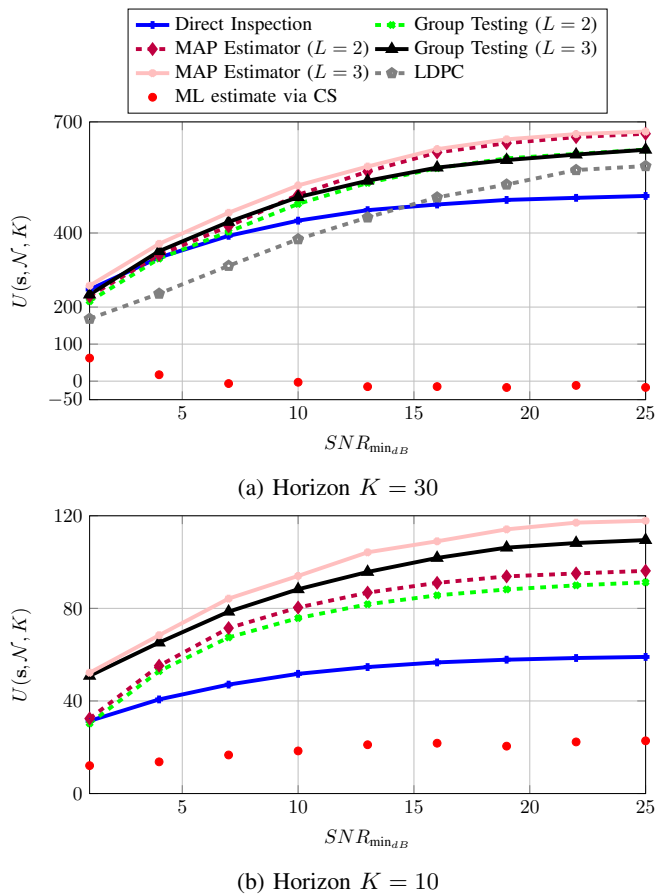

\centering
\begin{subfigure}{\linewidth}
\centering
\includegraphics[width=\linewidth]{Figures/K_30_SNR_vec_empty.tikz}
\caption{Horizon $K=30$}
\end{subfigure}
\begin{subfigure}{\linewidth}
\centering
\includegraphics[width=\linewidth]{Figures/K_10_SNR_vec_empty.tikz}
\caption{Horizon $K=10$}
\end{subfigure}
\caption{Comparing utility for different approaches vs $SNR_{\text{min}}$ for different horizons $K$ and the prior $\omega_i\sim\mathcal{U}\left(0.7,\frac{\rho_i}{\rho_i+r_i}\right)$ (number of resources $N=20$)}\label{fig:SNR_vec}
\end{figure}
{\color{black}

\subsubsection{Detection performance vs MWC}
In this last section, we compare our approach to the performances of another incoherent power spectrum sensing approach proposed in \cite{cohen2014sub}, which was mentioned in Section \ref{sec:MWC}. We considered $150$ independent bands with an expected $5\%$ occupancy, e.g. $\omega_i=0.95~~\forall i\in\N$. For the MWC scheme there are $M=30$ analog channels (in expectation twice the support of the double-sided bandwidth occupancy). We recall that our definition of SNR is the worst case per sub-band when the signal is present, e.g. $SNR=\frac{\varphi_{\text{min}}}{n}$, where we considered for simplicity $n_i=n~\forall i\in\N$, whereas the SNR considered in \cite{cohen2014sub} is $\frac{\sum_{i\in\N}s_i\varphi_i}{Nn}$. As indicated in \cite{cohen2014sub}, for a fixed sensing time it is preferable to choose approximately the same number of frames and samples per frame. This is the case in the simulations in Fig. \ref{fig:MWC_comp}, where $\kappa$ indicates the product of the two quantities $KP$. For the utility parameters, in light of Assumption \ref{ass:no_sensing_no_utility}, we considered $r_i=1$ and $\rho_i=19, ~\forall i\in\N$. We remark however, that the ROC curves in Fig.\ref{fig:MWC_comp} for the MWC power spectrum sensing, do not depend on these parameters: the different points in the curve are obtained by changing $\lambda$ ($\lambda_i=\lambda~\forall i\in\N$) in \eqref{eq:LASSO}. 
Notice that, for the same sampling frequency in each branch, the scheme in \cite{cohen2014sub} collects $M$ times the observations we collect per unit of time, hence the two points for the optimized strategy correspond to: 1) an unfair comparison with our scheme, where we keep $\frac{1}{M}$ observations (indicated with $R_s$), and 2) a fair comparison where we assume our scheme can collect the same number of observations, sampling at $M\cdot{R_s}$, hence obtaining a factor of $M$ SNR gain per test\footnote{We remind the reader that the sampling rate in \cite{cohen2014sub} should actually be higher than $R_s$, as discussed in Section \ref{sec:MWC}. }. We note that at $10~dB$  (Fig.\ref{subfig:MWC_comp_10dB}), the MWC scheme needs approximately $8\cdot M=240$ times more observations to outperform our proposed approach in the unfair comparison, while at $20~dB$ our approach offers better detection performances, even when we let MWC collecting $200\cdot M$ more observations. What is  remarkable is how much higher is the gain of our approach at higher SNR. This is due to the effect of the $SNR$ on the covariance estimate, which is required in the power spectrum sensing algorithm in \cite{cohen2014sub}.
In fact, while extending the number of samples per frame $K$ can mitigate the noise-folding problem, improving the accuracy of the covariance estimate requires a larger number of frames $P$, despite good $SNR$. For instance, for a Gaussian random variable with zero mean and variance $\sigma^2$, it is relatively straightforward to find that the ML estimate for the variance, has itself variance equal to $\frac{2\sigma^4}{P}$ (for $P$ observations). This implies that the LASSO-recovery step does not keep to improve for higher $SNR$, but rather by averaging more, i.e., for higher $P$. 
}
\begin{figure}[t]
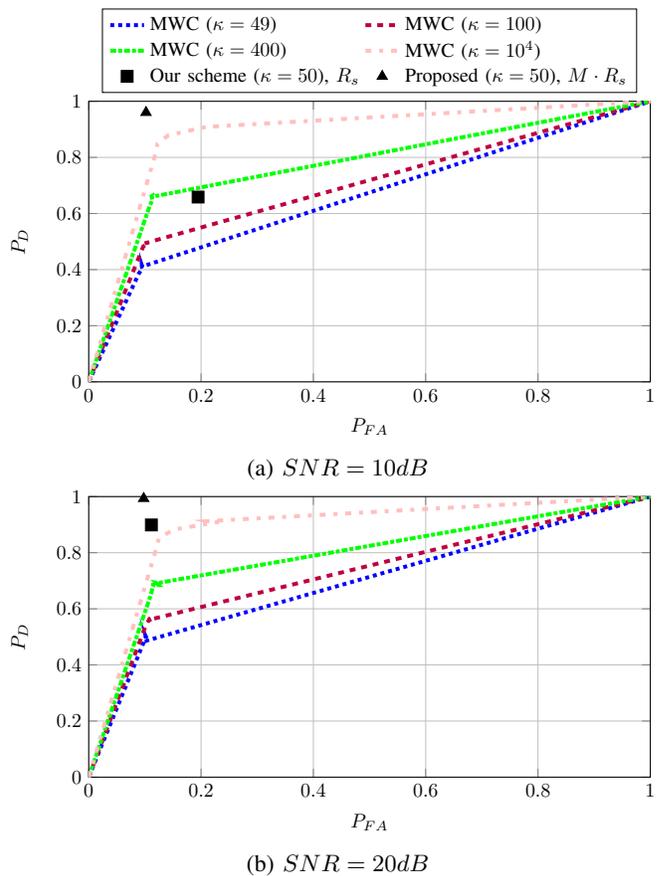

\centering
\begin{subfigure}{\linewidth}
\includegraphics[width=\linewidth]{Figures/MWC_comparison_10dB.tikz}
\caption{$SNR=10dB$}\label{subfig:MWC_comp_10dB}
\end{subfigure}
\begin{subfigure}{\linewidth}
\includegraphics[width=\linewidth]{Figures/MWC_comparison_20dB.tikz}
\caption{$SNR=20dB$}\label{subfig:MWC_comp_20dB}
\end{subfigure}
\caption{{\color{black} ROC curve for the power spectrum sensing algorithm in \cite{cohen2014sub} with MWC front-end, and comparison with the detection performances of our optimized sensing strategy}}\label{fig:MWC_comp}
\end{figure}
\section{Conclusions}
In this work, we proposed a new framework to optimize the performances of an opportunistic spectrum access strategy with sub-Nyquist sampling, and described the connection between such strategy and the design of the front-end sampling architecture. 
For the problem proposed, we characterized the factor approximation of the greedy strategy and showed, via numerical results, how our dynamic design for the sensing matrix guarantees better performances than other static approaches, namely: the linearization of an ML estimate using a dense CS-sensing matrix, Belief Propagation using a regular Low-Density Parity-Check sensing matrix, and the Xampling power spectrum sensing strategy proposed in \cite{cohen2014sub}. 
\begin{appendix}
\begin{figure}
\centering
\includegraphics[width=\linewidth]{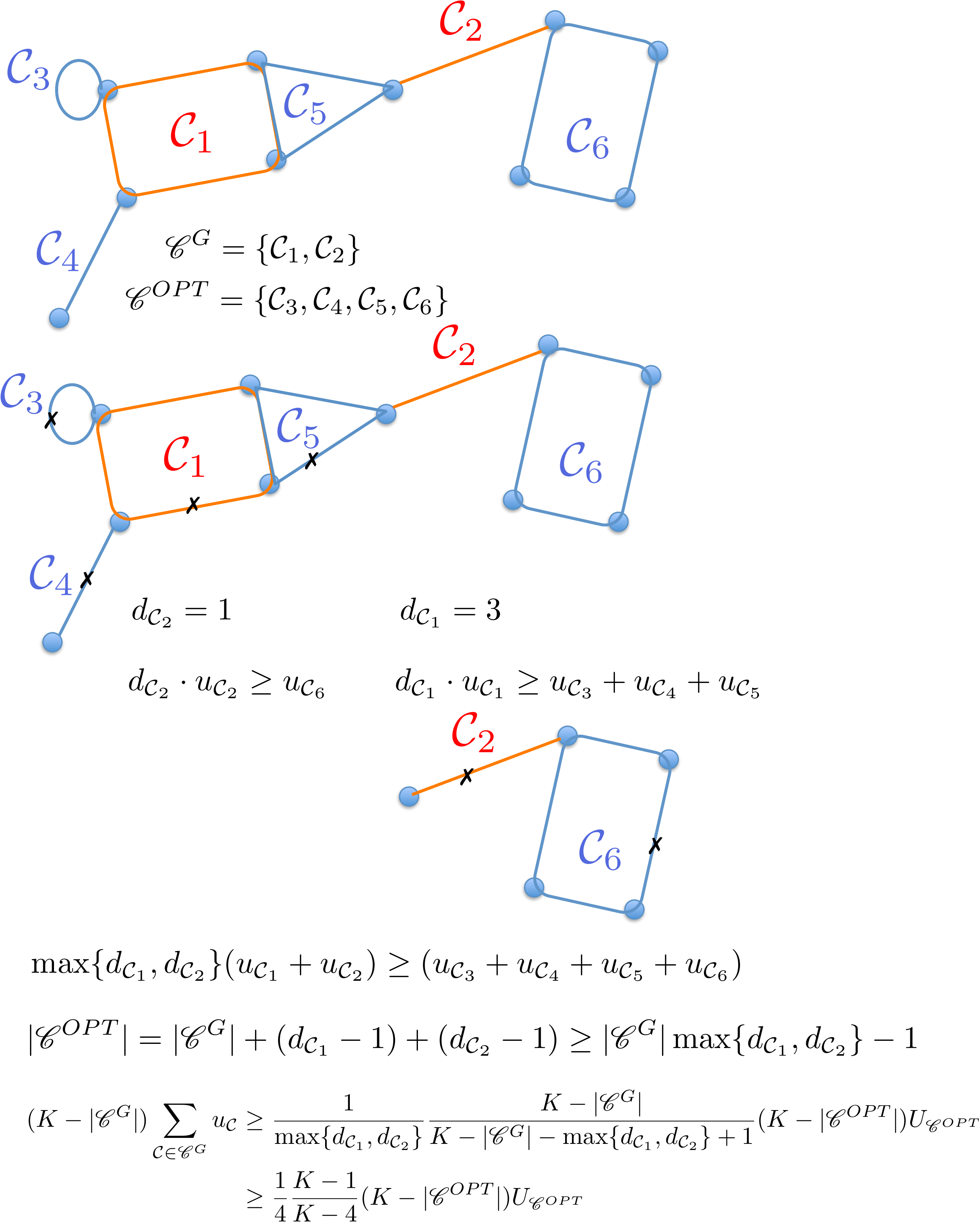}
\caption{Representation of the iterative procedure to obtain the factor approximation of the greedy algorithm}\label{fig:approximation_factor}
\end{figure}
\subsection{Proof of Lemma \ref{lemma:DI_submodular}}\label{app:proof_DI_submodular}
To prove the submodularity of $U^{DI}(\A)$ we show that, to prove the submodularity property: 
\be\label{eq:submodularity_property}
U^{DI}(\A+a)+U^{DI}(\A+b)\geq U^{DI}(\A+a+b)+U^{DI}(\A).
\ee  
is equivalent to prove
\be 
u_a^{DI}+u_b^{DI}\geq 0
\ee  
which is true by assumption on the function $u_i^{DI}$ defined in \eqref{eq:U_i_def}.
Since $u_i^{DI}=0$ for $\alpha_i=1,\beta_i=0$,  we have $u_i^{DI}\geq 0 ~~\forall i\in\N$ for the optimized $\alpha_i^{DI},\beta_{i,\text{max}}^{DI}$.
\subsection{Proof of Lemma \ref{lemma_submodular_mixed}}\label{app:proof_submodular_mixed}
The function is the sum of two terms, to prove the first one is sub-modular one can follow the same steps in Appendix \ref{app:proof_DI_submodular}.
For the second term, it is enough to show that, for any $i$, $-\Upsilon(\deg_{\E}(i))$ is a sub-modular function of $\E$. 
The function is clearly sub-modular since is a concave function of the nodal degree, and from this we can conclude the second term is a positive sum of sub-modular functions, hence sub-modular. 
To prove the equivalence of the two optimizations in \eqref{eq:submodular_problem_L_2_R_1_initial}-\eqref{eq:submodular_problem_L_2_R_1}, we first note that for any $\E$ that satisfies the constraints in \eqref{eq:submodular_problem_L_2_R_1_initial}, the second term of the objective in \eqref{eq:submodular_problem_L_2_R_1} is equal to 0 and the two objectives are equal.
It follows that we simply need to verify that no set of edges, that violate the constraint on the nodal degree, would be the optimal solution for \eqref{eq:submodular_problem_L_2_R_1}. To show this, we note that any infeasible set of edges (according to \eqref{eq:submodular_problem_L_2_R_1_initial}) can be made feasible by removing some edges. For $M$ large enough, i.e. $M>K\max_{ij}u_{ij}$ it is relatively straightforward to verify that such removal of edges would improve the objective, preventing an infeasible solution for \eqref{eq:submodular_problem_L_2_R_1_initial} to be optimal for \eqref{eq:submodular_problem_L_2_R_1}, and this concludes the proof.
\subsection{Proof of Lemma \ref{lemma_approximation_factor}}\label{app:proof_approximation_factor}
We want to prove 
\be\label{eq:what_we_want_to_prove}
U^{GT}\left(\mathscr{C}^G\right)\geq \alpha U^{GT}(\mathscr{C}^{OPT})
\ee 
where $\alpha=\frac{1}{\min\{L_{\text{eff}},\frac{K}{2}\}}\frac{K-1}{K-\min\{L_{\text{eff}},\frac{K}{2}\}}$ and $L_{\text{eff}}\leq L$ is the largest test size returned by the greedy algorithm.
We also rewrite 
\begin{align*}
U^{GT}\left(\mathscr{C}^G\right)&=(K-|\mathscr{C}^G|)U_{\mathscr{C}^G}\\
U^{GT}(\mathscr{C}^{OPT})&=(K-|\mathscr{C}^{OPT}|)U_{\mathscr{C}^{OPT}}
\end{align*}
To prove the claim we look at the graph obtained by the union of the cycles in the optimal and the greedy solution. 
Since in each of the solutions, no node can be in two cycles, it follows that in the obtained graph, no node can be in more than two cycles. 
Let us start by assuming there is a cycle $\C$ with associated utility $u_{\C}$ in the optimal solution that does not share any node with the greedy solution. 
This means 
\begin{align}
&\frac{U_{\mathscr{C}^{OPT}}-u_{\C}}{K-|\mathscr{C}^{OPT}|}\leq u_{\C}\leq \frac{U_{\mathscr{C}^G}}{K-|\mathscr{C}^G|-1}\label{eq:first_condition_no_isolated_nodes}\\
&\rightarrow   U_{\mathscr{C}^{OPT}}\leq (K-|\mathscr{C}^{OPT}|+1)u_{\C}\leq U_{\mathscr{C}^G} \frac{K-|\mathscr{C}^{OPT}|}{K-|\mathscr{C}^G|}\label{eq:second_condition_no_isolated_nodes}
\end{align}
where \eqref{eq:first_condition_no_isolated_nodes} follows from the fact that adding $\C$ to $\mathscr{C}^{OPT}\setminus \{\C\}$ improves the objective, but would not improve the objective for the greedy solution.
From \eqref{eq:second_condition_no_isolated_nodes} we could then conclude $|\mathscr{C}^{G}|>|\mathscr{C}^{OPT}|$, since for $|\mathscr{C}^{G}|\leq |\mathscr{C}^{OPT}|$ we would find from \eqref{eq:second_condition_no_isolated_nodes} that $U^{GT}\left(\mathscr{C}^G\right)\geq U^{GT}(\mathscr{C}^{OPT})$.
This means we can replace a cycle in $\mathscr{C}^G$ with this isolated cycle, to form a set of cycle $\tilde{\mathscr{C}}^G$ whose objective is lower than $\mathscr{C}^G$ by greedy search. In light of \eqref{eq:second_condition_no_isolated_nodes}, we can iterate this process by always picking the cycle to be replaced, in such a way that all the cycles in the optimal solution share at least one node with the set of cycles in $\tilde{\mathscr{C}}^G$.
Now we have that all the cycles in the optimal solution share at least one node with a cycle in $\tilde{\mathscr{C}}^G$. 
If, instead, one has that no cycle $\C$ in the optimal solution is isolated, and that there are isolated cycles in the greedy solution, then the set $\tilde{\mathscr{C}}^G$ is formed by removing these cycles from $\mathscr{C}^G$, lowering the objective (by submodularity and greedy search). One would then again obtain that all the cycles in the optimal solution share at least one node with the set of cycles in $\tilde{\mathscr{C}}^G$.
We now iteratively remove cycles from $\tilde{\mathscr{C}}^G$ and $\mathscr{C}^{OPT}$, while bounding the loss in performance and therefore obtain the factor approximation we want to prove. 
We can remove cycles from $\tilde{\mathscr{C}}^G$ in decreasing order of utility and since we know that for each cycle $\C'$ of length $L$ in $\tilde{\mathscr{C}}^G$ there are at most $L$ different cycles in the optimal solution that share a node with $\C'$, by greedy search we have that $L\cdot u_{\C'}$ is greater than the utility given by the cycles in the optimal solution that are adjacent to cycle $\C'$. 
Let us then define $\hat{\mathscr{C}}^G$ as the minimal subset of $\tilde{\mathscr{C}}^G$ that can cause the removal of all the cycles in the optimal solution when iterating the procedure just described, i.e. the set containing the first $|\hat{\mathscr{C}}^G|$ in decreasing order of utility contained in $\tilde{\mathscr{C}}^G$.
Again, by sub-modularity and greedy search, one can easily find that the objective for $\hat{\mathscr{C}}^G$ is lower than $\tilde{\mathscr{C}}^G$, since if the objective could not be improved by removing a cycle from $\mathscr{C}^G$, then it also cannot improve the objective for $\tilde{\mathscr{C}^G}$, which has utility strictly greater than $\mathscr{C}^G$. 
At this point we can prove \eqref{eq:what_we_want_to_prove} for $\hat{\mathscr{C}}^G$and this will prove it for $\mathscr{C}^G$.
We use $d_{\C}$ to define the number of cycles in $\mathscr{C}^{OPT}$ that can be removed by removing the cycle $\C$ in $\hat{\mathscr{C}}^G$ and $d\triangleq\max\limits_{\C\in\hat{\mathscr{C}}^G}d_{\C}$. 
By iterating our procedure described above, we end up having
\begin{align*}
&(K-|\hat{\mathscr{C}}^G|)U_{\hat{\mathscr{C}}^G}=\frac{1}{d}\frac{K-|\hat{\mathscr{C}}^G|}{K-|\mathscr{C}^{OPT}|}(K-|\mathscr{C}^{OPT}|)d\cdot U_{\hat{\mathscr{C}}^G}\\
&\geq \frac{1}{d}\frac{K-|\hat{\mathscr{C}}^G|}{K-|\hat{\mathscr{C}}^G|-d+1}(K-|\mathscr{C}^{OPT}|)d\cdot U_{\hat{\mathscr{C}}^G}\\
&\geq \frac{1}{d}\frac{K-|\hat{\mathscr{C}}^G|}{K-|\hat{\mathscr{C}}^G|-d+1}(K-|\mathscr{C}^{OPT}|)U_{\mathscr{C}^{OPT}}\\
&\geq \frac{1}{\min\{L_{\text{eff}},\frac{K}{2}\}}\frac{K-1}{K-\min\{L_{\text{eff}},\frac{K}{2}\}}(K-|\mathscr{C}^{OPT}|)U_{\mathscr{C}^{OPT}}\\
\end{align*}
and since $(K-|{\mathscr{C}}^G|)U_{{\mathscr{C}}^G}\geq(K-|\hat{\mathscr{C}}^G|)U_{\hat{\mathscr{C}}^G}$, this concludes the proof.
We have used the fact that $d\leq L_{\text{eff}}$ and that the function $d(K-d)$ has its maximum in $d=\frac{K}{2}$. 
Fig.\ref{fig:approximation_factor} shows an example of the iterative procedure to obtain the bound just derived.
\end{appendix}
\bibliographystyle{IEEEtran}
\bibliography{Opportunistic_Sensing_Bibliography}
\vspace{-1.5cm}
\begin{IEEEbiography}
[{\includegraphics[width=1in,height=1.25in,clip,keepaspectratio]{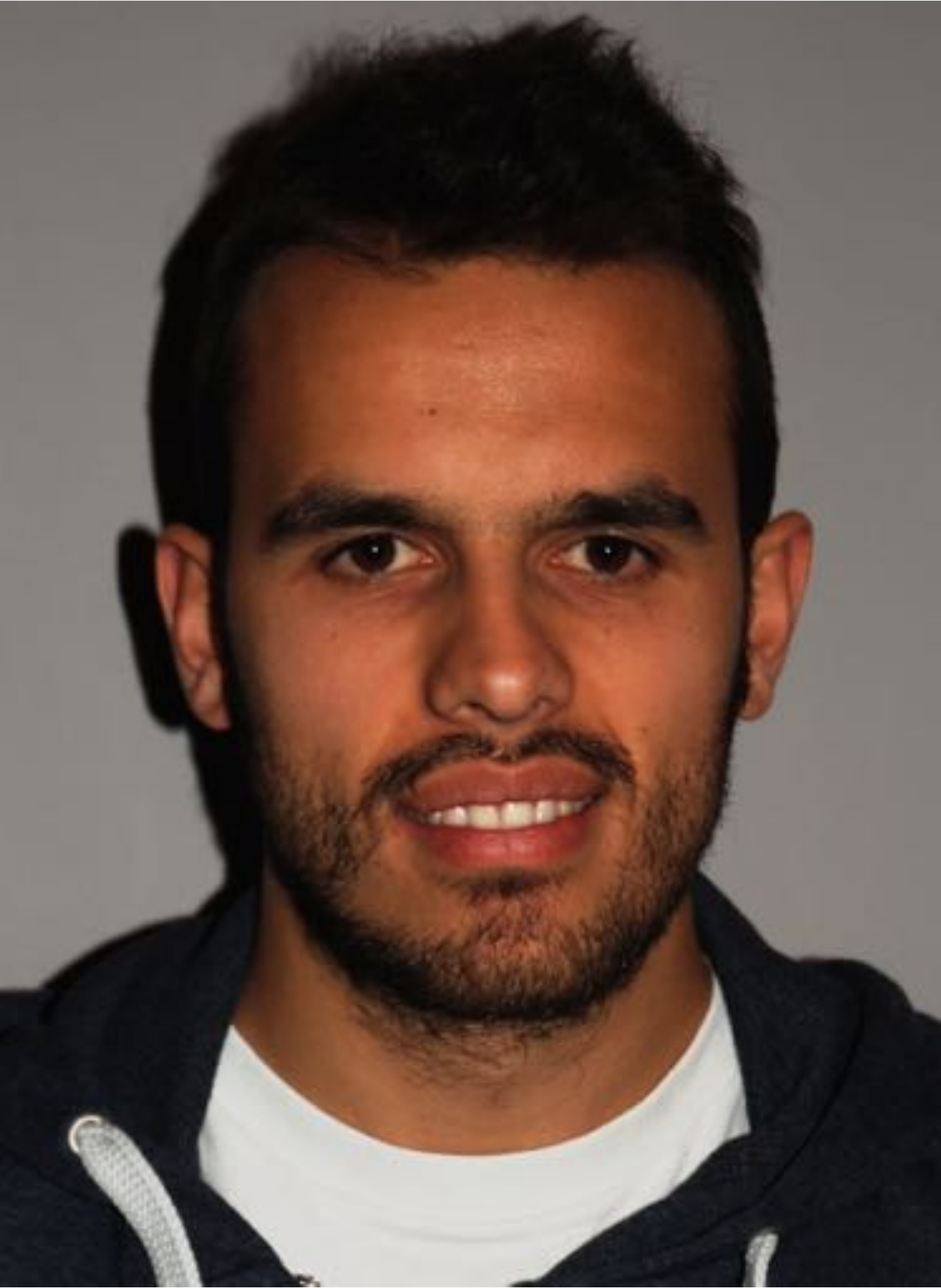}}]
{Lorenzo Ferrari} (S '14) received his PhD in Electrical Engineering at Arizona State University and is now a Senior Engineer at Qualcomm Research.
Prior to that, he received his B.Sc. and M.Sc degree in Electrical Engineering from University of Modena, Italy in 2012 and 2014 respectively. His research interests lie in the broad area of wireless communications and signal processing. He has received the IEEE SmartGridComm 2014 Best Student Paper Award for the paper ``The Pulse Coupled Phasor Measurement Unit''.
\end{IEEEbiography}
\vspace{-1.5cm}
\begin{IEEEbiography}
[{\includegraphics[width=1in,height=1.25in,clip,keepaspectratio]{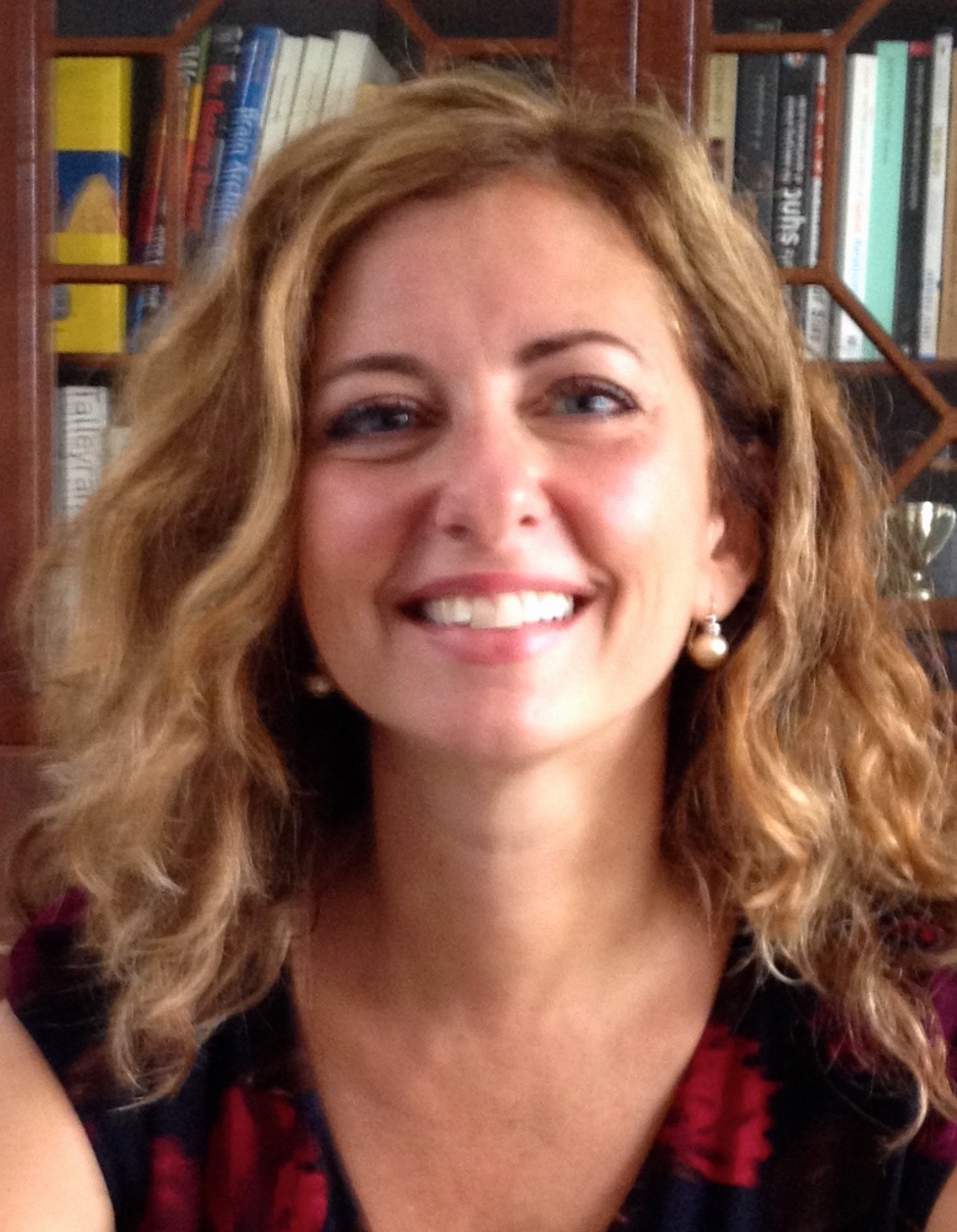}}]
{Anna Scaglione} (F '11) is currently a professor in
electrical and computer engineering at Arizona State
University. She was previously at the University of
California at Davis, Cornell University and University
of New Mexico. Her research focuses on various
applications of signal processing in network and data science
that include intelligent infrastructure for energy delivery and information
systems.
Dr. Scaglione was elected an IEEE fellow in
2011. She received the 2000 IEEE Signal Processing
Transactions Best Paper Award and more recently
was honored for the 2013, IEEE Donald G. Fink Prize Paper Award for the
best review paper in that year in the IEEE publications, her work with her
student earned 2013 IEEE Signal Processing Society Young Author Best Paper
Award (Lin Li). She was EIC of the IEEE Signal Processing Letters and
served in many other capacities the IEEE Signal Processing, IEEE
Communication societies and is currently Associate Editor for the 
IEEE Transactions on Control over Networked Systems.
\end{IEEEbiography}
\end{document}